\documentclass[12pt,reqno]{amsart}

\usepackage[margin=2.5cm]{geometry}
\usepackage{tikz}
\tikzset{node distance=2cm, auto}
\usetikzlibrary{matrix,arrows,decorations.pathmorphing,decorations.markings,calc}

\usepackage{graphicx}
\usepackage{caption}
\usepackage{rotating}
\usepackage{subfig}
\usepackage{float}
\usepackage{amsmath,amssymb}

\usepackage{amscd}
\usepackage{verbatim}
\usepackage{enumerate}
\usepackage{bbm}
\usepackage{mathrsfs}
\usepackage{color}

\newtheorem{proposition}{Proposition}
\newtheorem{lemma}{Lemma}
\newtheorem{remark}{Remark}

\newcommand{\w}{\widetilde}

\newcommand{\CC}{\mathbb{C}}

\newcommand{\pa}{\partial}

\newcommand{\be}{\begin{equation}}
\newcommand{\ee}{\end{equation}}

\newcommand{\x}{\times}

\title{Limit Shapes of the Stochastic Six Vertex Model}
\author{Nicolai Reshetikhin}
\address{N.R.: Department of Mathematics, University of California, Berkeley,
CA 94720, USA \\  \& KdV Institute for Mathematics, University of Amsterdam, 1098 XH Amsterdam, The Netherlands  \\ \& ITMO University. Saint Petersburg 197101, Russia.}
\email{reshetik@math.berkeley.edu}
\author{Ananth Sridhar}
\address{A.S.: Department of Physics, University of California, Berkeley,
CA 94720, USA}
\email{asridhar@berkeley.edu}
\begin{document}

\begin{abstract} It is shown that limit shapes for the stochastic 6-vertex model on a
cylinder with the
uniform boundary state on one end are solutions to the Burger type equation.
Solutions to these equations are studied for step initial conditions. When the circumference
goes to infinity the solution corresponding to critical initial densities  coincides with the
one found by Borodin, Corwin and Gorin.

\end{abstract}
\maketitle
\tableofcontents

\section*{Introduction}
The 6-vertex model is a lattice model of statistical mechanics with a long history, beginning with Pauling's study of the residual entropy of ice. The states of the model are orientations assigned to edges of the lattice satisfying the rule that each vertex has two incoming and two outgoing edges. States can equivalently be described as a ensembles of non-crossing lattice paths or as stepped surfaces called height functions.

The 6-vertex model is an exactly solvable lattice model; the weights of the model can be arranged to form commuting families of row-to-row transfer matrices that can be simultaneously diagonalized by Bethe ansatz \cite{LW}\cite{FadTakh}. The construction of transfer matrices and their algebraic properties are closely to the representation theory of the quantized universal enveloping algebra  $U_q(\widehat{sl_2})$, see for example \cite{JM} \cite{Res10}.

In the thermodynamic limit, the 6-vertex model exhibits the limit shape phenomena. The state of the system becomes deterministic at the macroscopic scale, with statistical fluctuations remaining only at the microscopic scale. Conjecturally, the height function in the thermodynamic limit can be computed via a variational principle  \cite{CpPr}\cite{CPZJ}\cite{PR} \cite{ZJ} \cite{RS}. Such variational principle was developed
for dimer models in \cite{CKP}.

For special values of Boltzmann weights in the 6-vertex model the transfer-matrix satisfies
Markov property and can be regarded as the collection of transition probabilities for a stochastic
process on a row of vertical edges of the model. For these values of weights the model is known as
the stochastic 6-vertex model \cite{BCG}.  The corresponding Markov process is a discrete time generalization of the asymmetric exclusion process (ASEP).

The goal of this paper is to study the limit shapes of the stochastic 6-vertex model on the cylinder. Our main result is  the partial differential equations determining the height function of the stochastic 6-vertex model. This result is derived from the conjectural variational principle for the 6-vertex model.
The limit shape with the uniform boundary state on the top end of the cylinder turns out to be a first order differential equations of the Burger's type, which can be solved by the method of characteristics. In the ASEP limit of the 6-vertex model, we recover the inviscid Burgers equation as computed by \cite{BF} \cite{Rez}\cite{GS}. As an example, we compute the height function exactly for domain wall type boundary conditions on the cylinder. When the initial condition is a step function, our solution converges to the limit shape found in \cite{BCG} in the limit when the circumference of the cylinder goes to infinity.

There are a few areas of further work that we plan to pursue in followup publications. Firstly, note (see Section \ref{sec:stochastic}) that the uniform boundary state on the top of the cylinder is crucially a co-eigenstate of the transfer matrix with the eigenvalue $1$, corresponding to the Markov property of the stochastic model. Similar results to those in this paper hold for the general 6-vertex model, with the ground co-state being the boundary state for the top end of the
cylinder.
Secondly, for the generic phases of the 6-vertex and dimer model, the statistical fluctuations about the limit shape are conjecturally described by the Gaussian free field determined by surface tension functional near the limit shape. The degeneracy of the surface tension function at the stochastic point (see Lemma \ref{zero-He}) explains the appearance of the Tracy-Widom distribution and KPZ scaling as computed in \cite{BCG}. Fluctuations for TASEP
model near shocks were described recently in \cite{BL}. Fluctuations of height functions in the stochastic
6-vertex model were studied recently in \cite{A}, where, among other results the convergence of these fluctuations to Baik-Rains distribution was proven. It has been shown there that fluctuations are Airy along
a special line. From our point of view this is the line where the quadratic variation
of the "effective action" (\ref{eff-action}) is non-degenerate. Another interesting results about limit shapes were obtained in \cite{CS} where the authors studied tangent lines to limit shapes. These
tangent lines may have relation to the characteristics of the Burgers type
equation obtained in this paper.
Lastly, we expect that similar results hold for the higher spin 6-vertex models.

The outline of the paper is as follows: in the first section we recall basic facts about the 6-vertex model with magnetic fields. In the second
section we describe the stochastic weights and the limit to the ASEP. The third section contains the derivation of the Burgers type differential equation for the the limit shape with the uniform distribution on
one end of the cylinder. In the last section we give some solutions for the step boundary conditions
on the other end of the cylinder. The Appendix A contains the prove of the stochasticity of transfer matrices on a cylinder for stochastic weights. The Appendix B contains necessary facts about the free energy
of the 6-vertex model.

Acknowledgements: We would like to thank D. Keating for many helpful discussions and for running numerical simulations. Both authors were supported by the NSF  grant DMS-1201391. The work of A.S. was supported by RTG NSF grant 30550. We are grateful for the hospitality at Universite Paris VII.
The research of N.R. was partly supported
by the Russian Science Foundation (project no. 14-11-00598), he also would like to thank QGM center at the University of Aarhus for hospitality.

\section{The 6-vertex model}
In this paper, we will focus primarily on the 6-vertex model defined on the square lattice embedded on the cylinder. In this section we review the essential definitions and refer the reader to \cite{PR},\cite{AR},\cite{Ba} for further details.

\subsection{The 6-vertex model on the cylinder}

By the \emph{square lattice} $ \mathbb{Z}^2 $ we mean the graph embedded in $ \mathbb{R}^2 $ with vertices at positions $ ( m, n) \in \mathbb{R}^2 $ for $ m,n \in \mathbb{Z} $ and edges joining all vertices distance $ 1 $ apart.

\subsubsection{Cylinder Graphs}
The \emph{cylinder graph} $ \mathcal{C}_{MN} $ is the quotient of the subgraph $ \mathbb{Z} \times \{0 \cdots N-1 \} \subset \mathbb{Z}^2 $ by $ (m,n) \sim ( m+ M, n) $.

The graph $ \mathcal{C}_{MN} $ consists internal vertices, internal edges, boundary vertices, and boundary edges. The sets of \emph{boundary vertices}, denoted by $ \partial V_1 $ and $\partial V_2 $,  consist of vertices at positions $ (m,0) $ and $ (m,N-1) $ respectively. The \emph{internal vertices} are non-boundary vertices. The \emph{boundary edges} $ \partial E_{i} $  consists of edges adjacent to exactly one boundary vertex in $ \partial V_{i} $, and the \emph{internal edges} are those adjacent to only internal vertices. The \emph{boundary faces} $ \partial F_{i} $ are those adjacent to boundary edges in $ \partial E_{i} $.

\subsubsection{Configurations and Boltzmann Weights}
A \emph{configuration} $ S $  of the 6-vertex model is a subset of internal and boundary edges satisfying the ice rule: at each internal vertex $ v $, the edges in $ S $ adjacent $ v $ must be one of six possible local configurations as shown in Figure \ref{fig:icerule}. The ice rule implies that a configuration of the 6-vertex model can be seen as an ensemble of paths that do not cross (although they can touch at an $ a_2 $ vertex), see Figure \ref{fig:cylindergraph}.

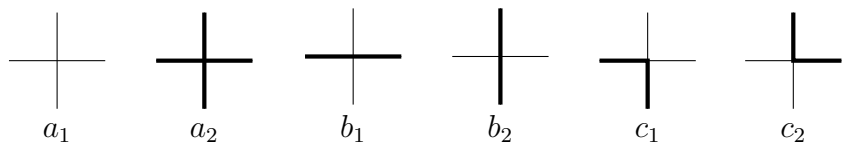
\begin{figure}[h]
\begin{tikzpicture}[scale=.8]
\draw (-.8,0) -- (.8,0);
\draw (0,-.8) -- (0,.8);
\node at (0,-1.2){$a_1$};
\end{tikzpicture} \; \;
\begin{tikzpicture}[scale=.8]
\draw (-.8,0) -- (.8,0);
\draw (0,-.8) -- (0,.8);
\draw[ultra thick](-.8,0)--(.8,0);
\draw[ultra thick](0,-.8)--(0,.8);
\node at (0,-1.2){$a_2$};
\end{tikzpicture} \; \;
\begin{tikzpicture}[scale=.8]
\draw (-.8,0) -- (.8,0);
\draw (0,-.8) -- (0,.8);
\draw[ultra thick](-.8,0)--(.8,0);
\node at (0,-1.2){$b_1$};
\end{tikzpicture} \; \;
\begin{tikzpicture}[scale=.8]
\draw (-.8,0) -- (.8,0);
\draw (0,-.8) -- (0,.8);
\draw[ultra thick](0,-.8)--(0,.8);
\node at (0,-1.2){$b_2$};
\end{tikzpicture} \; \;
\begin{tikzpicture}[scale=.8]
\draw (-.8,0) -- (.8,0);
\draw (0,-.8) -- (0,.8);
\draw[ultra thick](0,-.8)--(0,0)--(-.8,0);
\node at (0,-1.2){$c_1$};
\end{tikzpicture} \; \;
\begin{tikzpicture}[scale=.8]
\draw (-.8,0) -- (.8,0);
\draw (0,-.8) -- (0,.8);
\draw[ultra thick](.8,0)--(0,0)--(0,.8);
\node at (0,-1.2){$c_2$};
\end{tikzpicture}

\caption{Six vertex configurations and their weights.} \label{fig:icerule}
\end{figure}

\begin{figure} 
\begin{tikzpicture}[xscale = .4, yscale = .5]
\fill (-4,0) circle (0pt);
\fill (4,0) circle (0pt) ;
\draw [ultra thin,  lightgray, domain = 0:360, smooth] plot ( { 2*sin(\x) },{3+ .6*cos(\x)} );
\draw [ultra thin,  lightgray, domain = 90:270, smooth] plot ( { 2*sin(\x) },{-2 + .6*cos(\x)} );
\draw [ultra thin,  lightgray] (2,3)--(2,-2);
\draw [ultra thin, lightgray] (-2,3)--(-2,-2);
\draw (2,-2)-- (2,3);
\draw (1.24698, -2.4691)-- (1.24698, 2.5309);
\draw (-0.445042, -2.58496)--(-0.445042, 2.41504);
\draw [dashed]  (-1.80194, -2.26033)--(-1.80194, 2.73967);
\draw [domain = 90:270, smooth] plot ( { 2*sin(\x) },{0+ .6*cos(\x)} );
\draw [domain = 90:270, smooth] plot ( { 2*sin(\x) },{1+ .6*cos(\x)} );
\draw [domain = 90:270, smooth] plot ( { 2*sin(\x) },{2 + .6*cos(\x)} );

\draw [ultra thick] (-0.445042, -2.58496)--(-0.445042, 2.41504);
\draw [ultra thick]  (-1.80194, -1.26033)--(-1.80194, 1.73967-2);
\draw [ultra thick] (1.24698, -3.4691+3)-- (1.24698, 2.5309-2);

\draw [ultra thick, domain = 90:141.429, smooth] plot ( { 2*sin(\x) },{0 + .6*cos(\x)} );
\draw [ultra thick, domain = 141.429:270, smooth] plot ( { 2*sin(\x) },{1 + .6*cos(\x)} );
\draw [ultra thick, domain = 90:244.286, smooth] plot ( { 2*sin(\x) },{-1 + .6*cos(\x)} );
\draw [ultra thick, domain = 244.286:270, smooth] plot ( { 2*sin(\x) },{  .6*cos(\x)} );
\draw [ultra thick] (2,2)-- (2,3);
\end{tikzpicture} \;
\begin{tikzpicture}
\draw[step=.5 cm, very thin] (-1.99,-1.48) grid (1.49,.98);
\draw[ thin, lightgray, decoration={markings,mark=at position 1 with {\arrow[scale=2]{>}}},
    postaction={decorate}](-2,-1.5) -- (-2,1.5);
\draw[ thin, lightgray, decoration={markings,mark=at position 1 with {\arrow[scale=2]{>}}},
    postaction={decorate}](-2,-1.5) -- (2,-1.5);

\draw[dashed] (1.5,-1.5)--(1.5,1);
\draw[dashed] (-2,-1.5)--(-2,1);

\draw[ultra thick] (-1.5,-1.5)--(-1.5,0)--(-2,0); \draw[ultra thick] (1.5,0)--(.5,0)--(.5,.5)--(-.5,.5)--(-.5,1);
\draw[ultra thick] (.5,-1.5)--(.5,-1)--(-2,-1)--(-2,-.5); \draw[ultra thick] (1.5,-1)--(1.5,-.5)--(1,-.5)--(1,1);
\draw[ultra thick] (0,-1.5)--(0,-.5)--(-1,-.5)--(-1,0)--(-1.5,0)--(-1.5,1);
\end{tikzpicture}
 
\caption{A six vertex configuration on a cylindrical graph.} \label{fig:cylindergraph}
\end{figure}
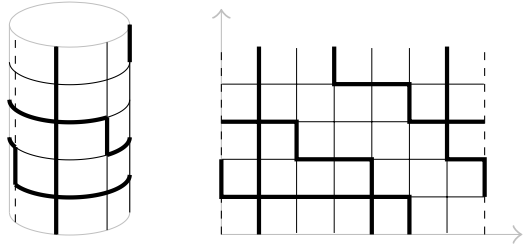
Each vertex $ v $ is assigned a vertex weight $ w(v,S) $ depending on the configuration of adjacent edges according to Figure \ref{fig:icerule}. The Boltzmann weight of $ S $ is
\begin{align*}
W (S) = \prod_{ v \in \text{int}(V)} w(v,S)
\end{align*}

\subsubsection{Magnetic Fields} \label{sec:magfield}We will assume the following useful parametrization of weights,  given in terms of $ (a,b,c) $, \emph{magnetic field} $ (H,V) $, and $ \lambda $ as:
\begin{equation} \label{eq:abchv}
\begin{aligned}
\begin{array}{cc}
 a_1 = a \; e^{H+V}  &  a_2 = a \; e^{-H - V} \\
 b_1 = b \; e^{-H+V}   &  b_2 = b \; e^{H-V} \\
 c_1 = c \; \lambda    &  c_2 = c \; \lambda^{-1}
\end{array}
\end{aligned}
\end{equation}
The ice rule implies that on the cylinder, the $ c_1 $ and $ c_2 $ type vertices occur in pairs, and there is no loss in generality by taking their weights to be the same, i.e. $\lambda=1$. However, in the stochastic
6-vertex model $a_i=b_i+c_i$ (see equation \ref{eq:stochasticweights}), which means it is convenient to keep $\lambda\neq 1$.

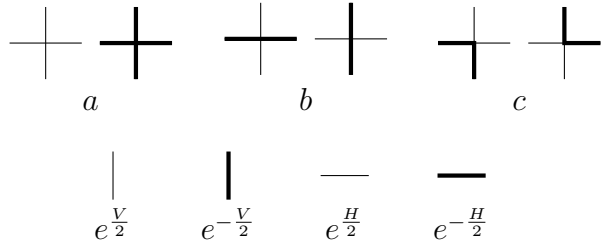
\begin{figure}[h]
\begin{tikzpicture}[scale=.8]
\begin{scope}[shift={(-.75,0)}]
\draw (-.6,0) -- (.6,0);
\draw (0,-.6) -- (0,.6);
\end{scope}
\node at (0,-1){$ a $};
\begin{scope}[shift={(.75,0)}]
\draw (-.6,0) -- (.6,0);
\draw (0,-.6) -- (0,.6);
\draw[ultra thick](-.6,0)--(.6,0);
\draw[ultra thick](0,-.6)--(0,.6);
\end{scope}
\end{tikzpicture} \; \;
\begin{tikzpicture}[scale=.8]
\begin{scope}[shift={(-.75,0)}]
\draw (-.6,0) -- (.6,0);
\draw (0,-.6) -- (0,.6);
\draw[ultra thick](-.6,0)--(.6,0);
\end{scope}
\node at (0,-1){$ b $};
\begin{scope}[shift={(.75,0)}]
\draw (-.6,0) -- (.6,0);
\draw (0,-.6) -- (0,.6);
\draw[ultra thick](0,-.6)--(0,.6);
\end{scope}
\end{tikzpicture} \; \;
\begin{tikzpicture}[scale=.8]
\begin{scope}[shift={(-.75,0)}]
\draw (-.6,0) -- (.6,0);
\draw (0,-.6) -- (0,.6);
\draw[ultra thick](-.6,0)--(0,0)--(0,-.6);
\end{scope}
\node at (0,-1){$ c $};
\begin{scope}[shift={(.75,0)}]
\draw (-.6,0) -- (.6,0);
\draw (0,-.6) -- (0,.6);
\draw[ultra thick](.6,0)--(0,0)--(0,.6);
\end{scope}
\end{tikzpicture} \vspace{10pt} \\
\begin{tikzpicture}[scale=.8]
\draw (0,-.4) -- (0,.4);
\node at (0,-.8){$ e^{\frac{V}{2} }$};
\end{tikzpicture} \; \;
\begin{tikzpicture}[scale=.8]
\draw [ultra thick] (0,-.4) -- (0,.4);
\node at (0,-.8){$ e^{- \frac{V}{2} }$};
\end{tikzpicture} \; \;
\begin{tikzpicture}[scale=.8]
\draw (-.4,0) -- (.4,0);
\node at (0,-.8){$ e^{\frac{H}{2} }$};
\end{tikzpicture} \; \;
\begin{tikzpicture}[scale=.8]
\draw [ultra thick] (-.4,0) -- (.4,0);
\node at (0,-.8){$ e^{- \frac{H}{2}} $};
\end{tikzpicture} \; \;
 
\caption{The symmetric vertex weights and magnetic fields on edges.} \label{fig:magneticfields}
\end{figure}

\subsubsection{Boundary Conditions and the Partition Function} The \emph{boundary configurations} $ \partial S_i $ of a configuration $ S $ are the restrictions of $ S $ to the boundary edges $ \partial E_i $.  A \emph{Dirichlet boundary condition} for the 6-vertex model fixes the boundary configurations of $ S $. Let $ \eta $ and $ \eta' $ be boundary configurations of $ \partial E_1 $ and $ \partial E_2 $. The cylinder partition function with fixed Dirichlet boundary conditions is defined as:
\begin{align} \label{eq:partfunc}
Z_{MN} \big[\eta, \eta' \big] = \sum_{\substack{\text{configs. } S \\ \partial S_1 = \eta \\ \partial S_2 = \eta'}} W(S)
\end{align}

A pair of boundary conditions $ (\eta, \eta' ) $ is called  {\it admissible} if $ Z_{MN}[ \eta, \eta'] \neq 0 $. Let $ m(\eta) $ be the number of edges in a boundary configuration $ \eta $. Since for any configuration $ S $ the ice rules imply $ m(\partial S_1) = m(\partial S_2) $, $ (\eta, \eta') $ is admissible only if $ m(\eta) = m(\eta') $.

A \emph{boundary state} is a probability distribution on the space of boundary conditions. We will write $ b(\eta, \eta') $ for the probability of $ (\eta, \eta') $ in the state $ b $. A \emph{local boundary state} on the cylinder is a distribution on the space of configurations on one end. Two local boundary states  $ b_1,  b_2 $ corresponding to two ends $\pa E_1$ and $\pa E_2$  of the cylinder define the product boundary state for the cylinder by $ b(\eta, \eta') = b_1(\eta) b_2(\eta') $. The partition function for a cylinder with two local boundary states is
\[
Z_{MN}(b_1,b_2)=\sum_{\eta,\eta} b_2(\eta) \; Z_{MN}[\eta,\eta'] \; b_1(\eta')
\]
Here we use round parantheses as opposed to square brackets for the partition function with Dirichlet  boundary conditions \ref{eq:partfunc}.

Simple examples of boundary states are:
\begin{enumerate}[1)]
\item Pure states (Dirichlet boundary conditions) in which the distribution, say $b_1$  is supported at a single configuration $\eta_0$, $ b_1(\eta) = \delta_{\eta, \eta_0} $.
\item The uniform distribution on the upper end. In this case $ b_2(\eta')=2^{-M} $.
\end{enumerate}
In this paper we will be mostly interested in boundary conditions which are uniform on the top end and Dirichlet on the lower end.

The partition function $ Z_{MN}(b_1,b_2) $ defines a probability measure $ \mu $ on the set of 6-vertex configurations on $ \mathcal{C}_{MN} $, where the probability of a configuration $ S $ is:
\begin{align*}
\mu(S) = \frac{b_1(\partial S_1) \; b_2(\partial S_2) \; W(S)}{Z_{MN}(b_1,b_2)}
\end{align*}
For any random variable $ X $ on six vertex states (for example, the height function introduced in section \ref{sec:heightfunction}), we will write $ \mathbb{E}_{MN}[X] $ for the expected value of $ X $ with respect to this  measure.

A boundary state $b$ defines a vector in a Hilbert space $ \mathscr{H}_M \simeq (\mathbb{C}^2)^{\otimes M} $ as follows. A boundary configuration $ \eta $ of $ \partial E_1 $ is identified with the element
$\psi_\eta \in \mathscr{H}_M $ of the tensor product basis:
\begin{align*}
\psi_\eta = e_{\eta(1)} \otimes e_{\eta(2)} \cdots \otimes e_{\eta(M)}
\end{align*}
where $ i = 1, \cdots, M $ enumerate the boundary edges; $ \eta(i) = 1 $ if the $i$th edge is in $ \eta $ and $ \eta(i) = 0 $ otherwise; and $ \{e_0, e_1 \} $ is the standard basis for $ \mathbb{C}^2 $. A local boundary state $ b $ is identified with the vector $ \sum_{\eta} b(\eta) \psi_\eta $.

Let $ Z_{MN}: \mathscr{H}_M \rightarrow \mathscr{H}_M $ be the linear map with matrix elements given by the partition functions $ Z_{MN}(\eta, \eta') $. The partition function with boundary states $ b_1 $ and $ b_2 $ is the inner product:
\begin{align} \label{eq:zip0}
Z_{MN}(b_1, b_2) = \langle b_2, Z_{MN} \; b_1 \rangle
\end{align}
where $\langle b_2,  b_1 \rangle=\sum_\eta b_1(\eta)b_2(\eta)$.

Thus, $\mathscr{H}_M$ can be naturally regarded as the space of states corresponding to one end of the cylinder. Boundary states for the cylinder are vectors in $\mathscr{H}_M\otimes \mathscr{H}_M$. Strictly speaking the first factor should be $\mathscr{H}_M^*$ but since we have a scalar product we identify the
vector space with its dual.

It is clear that $ \mathscr{H}_M $ has the decomposition:
\begin{align*}
\mathscr{H}_M  = \oplus_{k = 0}^{M} \mathscr{H}_M^{(k)}
\end{align*}
where $ \mathscr{H}_M^{(k)} $ is the subspace spanned by boundary configurations with $ m(\eta) = k $. Because of the 6-vertex rule, $ Z_{MN} $  has the block decomposition:
\[
Z_{MN}=\oplus_{m=0}^M Z_{MN}^{(m)},  \; \; \; Z_{MN}^{(m)}:\mathscr{H}_M^{(m)}\to \mathscr{H}_M^{(m)}
\]
The restricted partition function is defined as:
\begin{align} \label{eq:zip}
Z_{MN}^{(m)}(b_1, b_2) = \langle b_2^{(m)}, Z_{MN}^{(m)} \; b_1^{(m)} \rangle
\end{align}
where $ b^{(m)} $ is the orthogonal projection of $ b $ to $ \mathscr{H}_M^{(m)} $.

\subsubsection{The Transfer Matrix}  The transfer matrix $ T_M:  \mathscr{H}_M\rightarrow \mathscr{H}_M $  is the partition function (\ref{eq:zip}) with one row, $ T_M = Z_{M1} $. It has a more explicit description as follows.

Define the matrix $R: \CC^2 \otimes \CC^2 \rightarrow  \CC^2 \otimes \CC^2$
with matrix elements:
\begin{align*}
R=
\begin{pmatrix}
a e^{H+V} & 0 & 0 & 0 \\
0 & b e^{-H+V} & c\lambda^{-1}& 0 \\
0 & c\lambda & b e^{-V+H} & 0 \\
0 & 0 & 0 & a e^{-H-V}
\end{pmatrix}
\end{align*}
in the standard ordered tensor product basis $ \big( e_0 \otimes e_0,  e_0 \otimes e_1,  e_1 \otimes e_0, e_1 \otimes e_1 \big) $ for $ \CC^2 \otimes \CC^2 $. We will use Baxter's (projective) parametrization of the
symmetric model with zero magnetic fields: \footnote{ We use the parametrization when $\Delta=(a_1a_2+b_1b_2-c_1c_2)/2\sqrt{a_1a_2b_1b_2}>1$. These are the only values of $\Delta$ which occur in stochastic the 6-vertex model.}
\[
a=\sinh(u+\eta), \ \ b=\sinh(u), \ \ c=\sinh(\eta)
\]
and notation:
\begin{align*}
R(u)=
\begin{pmatrix}
\sinh(u+\eta) & 0 & 0 & 0 \\
0 & \sinh(u) & \sinh(\eta)& 0 \\
0 & \sinh(\eta) & \sinh(u) & 0 \\
0 & 0 & 0 & \sinh(u+\eta)
\end{pmatrix}
\end{align*}
It is clear that in Baxter's parametrization with $\lambda=e^\alpha$
\[
R=\left( D^V\otimes D^{H+\alpha}\right) \; R(u) \; \left( D^V\otimes D^{H-\alpha} \right)
\]
where
\begin{align}\label{DH}
D^\beta=\begin{pmatrix}e^{\beta/2}&0\\0&e^{-\beta/2}\end{pmatrix}
\end{align}

Enumerating the tensor components of $ \CC^2 \otimes \mathscr{H}_M \simeq (\mathbb{C}^2)^{\otimes (M+1)} $ by $\{ 0, 1, \cdots M \}$, define $ R_{ij}: \CC^2 \otimes \mathscr{H}_M \rightarrow \CC^2 \otimes \mathscr{H}_M $  as the map that acts as $ R $ on the $ i $th and $j $th tensor components and identity on the rest. The transfer matrix is then the trace over the $0$th tensor component of the matrix product:
\begin{align*}
T_M = \text{Tr}_0 ( R_{10} R_{20} \cdots R_{M0} )
\end{align*}
which can also be written as
\[
T_M(u)=(D^{2V}_1\cdots D^{2V}_M)^N t_M(u)
\]
where
\[
t_M(u)=\text{Tr}_0(D^{2H}_0 R_{10}(u)D^{2H}_0 R_{20}(u) \cdots D^{2}H_0 R_{M0}(u))
\]
The Yang-Baxter equation for $R(u)$ together with the ice rule imply the commutativity of transfer matrices:
\[
[T_M(u), T_M(v)]=0
\]
for any $u,v$.

It is clear that the transfer matrix preserves the spaces $ \mathscr{H}_M^{(m)} $, and we write $ T_M^{(m)} $ for the restriction of the transfer matrix to $ \mathscr{H}_M^{(m)} $.
The partition functions  in terms of the transfer matrix is:
\begin{align*}
Z_{MN}^{(m)}(b_1,b_2) = \langle b_2^{(m)}, \left(T_M^{(m)}(u) \right)^N, b_1^{(m)} \rangle
\end{align*}

\subsubsection{The Height Function}\label{sec:heightfunction}
To each 6-vertex configuration $ S $ on $ \mathcal{C}_{MN} $, we associate a function $ h_S $ on faces of the graph $ [0,M] \times [0,N-1] $. It is defined as follows. At the left corner face, set $ h_S(0,0) =0 $. Then when moving upward or rightward, the height function changes by $ +1 $ if crossing an edge in $ S $ and $ -1 $ otherwise (see Figure \ref{fig:heightfundef}).

The ice rule implies that the  monodromy $ h_S(M,n) - h_S(0,n) $ is a independent of $ n $. Thus $ h_S $ defines a multivalued function\footnote{Multivalued here means that it is a function on a cylinder with the branch cut chosen as $(m,0); m=0,\cdots, M-1$, see Fig. \ref{fig:cylindercut}.} on the faces of the cylinder $ \mathcal{C}_{MN} $ called the height function of $ S $. Note that (up to an additive constant) the boundary configuration of $ S $ determines the values of $ h_S $ on boundary faces.

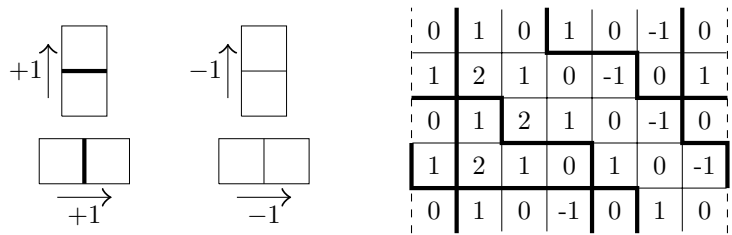
\begin{figure}[h]
\begin{tikzpicture}[scale=.6]
\draw[very thin] (-.5,0)--(-.5,2)--(.5,2)--(.5,0)--(-.5,0);
\draw[ultra thick] (-.5,1)--(.5,1);
\draw[decoration={markings,mark=at position 1 with {\arrow[scale=2]{>}}},
    postaction={decorate}] (-.8,.4) -- (-.8,1.6);
\node at (-1.3,1){\footnotesize{$+1$}};

\draw[very thin] (-1,-1.5)--(-1,-.5)--(1,-.5)--(1,-1.5)--(-1,-1.5);
\draw[ultra thick] (0,-1.5)--(0,-.5);
\draw[decoration={markings,mark=at position 1 with {\arrow[scale=2]{>}}},
    postaction={decorate}] (-.6,-1.8) -- (.6,-1.8);
\node at (0,-2.2){\footnotesize{$+1$}};
\end{tikzpicture} \; \; 
\begin{tikzpicture}[scale=.6]
\draw[very thin] (-.5,0)--(-.5,2)--(.5,2)--(.5,0)--(-.5,0);
\draw[very thin] (-.5,1)--(.5,1);
\draw[decoration={markings,mark=at position 1 with {\arrow[scale=2]{>}}},
    postaction={decorate}] (-.8,.4) -- (-.8,1.6);
\node at (-1.3,1){\footnotesize{$-1$}};
\draw[very thin] (-1,-1.5)--(-1,-.5)--(1,-.5)--(1,-1.5)--(-1,-1.5);
\draw[very thin] (0,-1.5)--(0,-.5);
\draw[decoration={markings,mark=at position 1 with {\arrow[scale=2]{>}}},
    postaction={decorate}] (-.6,-1.8) -- (.6,-1.8);
\node at (0,-2.2){\footnotesize{$-1$}};
\end{tikzpicture} \hspace{30pt}
\begin{tikzpicture}[scale=1.2]
\draw[step=.5 cm, very thin] (-1.99,-1.48) grid (1.49,.98);

\draw[dashed] (1.5,-1.5)--(1.5,1);
\draw[dashed] (-2,-1.5)--(-2,1);

\node at (-1.75,-1.25){\footnotesize{0}};
\node at (-1.25,-1.25){\footnotesize{1}};
\node at (-.75,  -1.25){\footnotesize{0}};
\node at (-.25,  -1.25){\footnotesize{-1}};
\node at (.25,   -1.25){\footnotesize{0}};
\node at (.75,   -1.25){\footnotesize{1}};
\node at (1.25, -1.25){\footnotesize{0}};

\node at (-1.75,-.75){\footnotesize{1}};
\node at (-1.25,-.75){\footnotesize{2}};
\node at (-.75,-.75){\footnotesize{1}};
\node at (-.25,-.75){\footnotesize{0}};
\node at (.25,-.75){\footnotesize{1}};
\node at (.75, -.75){\footnotesize{0}};
\node at (1.25, -.75){\footnotesize{-1}};

\node at (-1.75,-.25){\footnotesize{0}};
\node at (-1.25,-.25){\footnotesize{1}};
\node at (-.75,-.25){\footnotesize{2}};
\node at (-.25,-.25){\footnotesize{1}};
\node at (.25,-.25){\footnotesize{0}};
\node at (.75, -.25){\footnotesize{-1}};
\node at (1.25, -.25){\footnotesize{0}};

\node at (-1.75,.25){\footnotesize{1}};
\node at (-1.25,.25){\footnotesize{2}};
\node at (-.75,.25){\footnotesize{1}};
\node at (-.25,.25){\footnotesize{0}};
\node at (.25,.25){\footnotesize{-1}};
\node at (.75, .25){\footnotesize{0}};
\node at (1.25, .25){\footnotesize{1}};

\node at (-1.75,.75){\footnotesize{0}};
\node at (-1.25,.75){\footnotesize{1}};
\node at (-.75,.75){\footnotesize{0}};
\node at (-.25,.75){\footnotesize{1}};
\node at (.25,.75){\footnotesize{0}};
\node at (.75, .75){\footnotesize{-1}};
\node at (1.25, .75){\footnotesize{0}};

\draw[ultra thick] (-1.5,-1.5)--(-1.5,0)--(-2,0); \draw[ultra thick] (1.5,0)--(.5,0)--(.5,.5)--(-.5,.5)--(-.5,1);
\draw[ultra thick] (.5,-1.5)--(.5,-1)--(-2,-1)--(-2,-.5); \draw[ultra thick] (1.5,-1)--(1.5,-.5)--(1,-.5)--(1,1);
\draw[ultra thick] (0,-1.5)--(0,-.5)--(-1,-.5)--(-1,0)--(-1.5,0)--(-1.5,1);
\end{tikzpicture}
\caption{The height function rules and an example.} \label{fig:heightfundef}
\end{figure} 

The function $ h_S $ satisfies\footnote{The height function we defined here differs from the one we use in
\cite{RS} as $h_S=2\theta$, where $\theta$ is the discrete height function from \cite{RS}.}:
\begin{equation} \label{eq:lipshitzheight}
\begin{aligned}
 h_S(m+1,n) - h_S(m,n) &= \pm1 \\
h_S(m,n+1) - h_S(m,n) &= \pm 1
\end{aligned}
\end{equation}

A multivalued function $h'$ that satisfies (\ref{eq:lipshitzheight}) and $ h'(0,0) = 0 $, is called a \emph{height function} for the 6-vertex model. It is clear that there is a bijection between the 6-vertex configurations and height functions.

\subsection{The Thermodynamic Limit and Limit Shapes}

In this section, we recall basics about the 6-vertex model on the cylinder $ \mathcal{C}_{MN} $ in the limit $ M,N \rightarrow \infty $ with $ M / N $ fixed.

\subsubsection{Embedded Graphs and Normalized Height Functions}\label{subsec:nhf} Let $ C_{LT} = \mathbb{R} \times [0,T] \; /  \; \{ (x,y) \sim (x+L,y) \} $ be the flat cylinder of length $ T $ and circumference $ L $. We fix a branch cut $\{(0,y)\}_0^T$ when we discuss multivaued functions on $C_{LT}$, see Figure \ref{fig:cylindercut}.

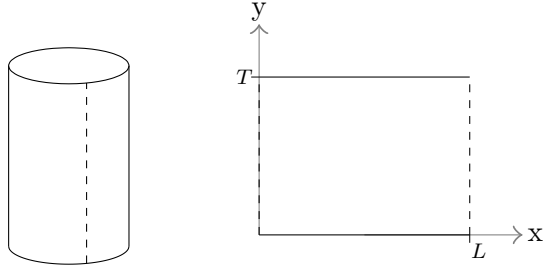
\begin{figure}[h]
\begin{tikzpicture}[baseline,scale = .4]
\fill (-4,0) circle (0pt);
\fill (4,0) circle (0pt) ;
\draw [domain = 0:360, smooth] plot ( { 2*sin(\x) },{3+ .6*cos(\x)} );
\draw [domain = 90:270, smooth] plot ( { 2*sin(\x) },{-3 + .6*cos(\x)} );
\draw (2,3)--(2,-3);
\draw (-2,3)--(-2,-3);
\draw [dashed] (0.59104,-3.5732)-- (0.59104,(2.4268);
\end{tikzpicture} \;
\begin{tikzpicture}[baseline, scale = .7]

\draw[thin, gray, decoration={markings,mark=at position 1 with {\arrow[scale=2]{>}}},
    postaction={decorate}](-2,-1.5) -- (-2,2.5);
\draw[ thin, gray, decoration={markings,mark=at position 1 with {\arrow[scale=2]{>}}},
    postaction={decorate}](0,-1.5) -- (3,-1.5);

\draw[dashed] (-2,-1.5)--(-2,1.5);
\draw[dashed] (2,-1.5)--(2,1.5);

\draw (-2,1.5)--(2,1.5);
\draw (-2, -1.5)--(2,-1.5);
\draw [ultra thin, dashed] (2,-1.5)--(2,-1.7);
\node at  (2.1,-1.8) { \tiny{ $L $}};
\draw [ultra thin, dashed] (-2,1.5)--(-2.2,1.5);
\node at  (-2.35,1.5) { \tiny{ $T $}};

\node at  (3.25,-1.5) { \small{x} };
\node at  (-2,2.75) { \small{y} };
\end{tikzpicture}
 
\caption{Branch cut on the cylinder.} \label{fig:cylindercut}
\end{figure}

Assuming $ L/T = M/N $, the graph $ \mathcal{C}_{MN} $ can be embedded in $ C_{LT} $ with mesh $ \epsilon = L/ N $ by $ \phi_\epsilon: \mathcal{C}_{MN} \hookrightarrow C_{LT} $, defined by rescaling coordinates $ (m,n) \mapsto (\epsilon m, \epsilon n ) $.

The embedding $\phi_\epsilon $ of the lattice to a plane, brings a height function $ h $ on $ \mathcal{C}_{MN} $ to a piecewise constant function on $ C_{LT} $, with constant value on each face of $ \phi_\epsilon(\mathcal{C}_{MN})\subset \mathbb{R}^2 $. We define the \emph{normalized height function} $\varphi $ on $ C_{LT} $ by $ \varphi(x,y) = \epsilon h\left(\lfloor x/ \epsilon \rfloor, \lfloor y/ \epsilon \rfloor \right) $. It satisfies constraints:
\begin{equation}
\begin{aligned}
\varphi(x + \epsilon,y) - \varphi(x,y)  = \pm \epsilon \\
\varphi(x , y + \epsilon) - \varphi(x,y) = \pm \epsilon
\end{aligned}
\end{equation}

In a similar way, a boundary configuration $ \eta $ on an end of the cylinder defines a piecewise constant function $ \chi: [0,L] \rightarrow \mathbb{R} $, satisfying $ \chi(x + \epsilon) - \chi(x) = \pm \epsilon $, called a \emph{normalized boundary height function}.

\subsubsection{Stabilizing Boundary Conditions and Thermodynamic Limit.} Let $\{\epsilon_n\}$ be a
sequence of positive numbers (meshes) such that $\epsilon_n \to 0$ as $ n \to \infty$.
Set $ N_n= \lfloor T/\epsilon_n \rfloor $ and $M_n= \lfloor L / \epsilon_n \rfloor $ where $T,L$ are the length and the circumference of the macroscopic cylinder $C_{LT}$.

Denote by $ H_L $ the space of functions $ \chi $ on $ [0,L] $, which satisfy the Lipshitz condition $ | \chi(x) - \chi(x') | < | x- x'| $, and denote the subspace of functions with fixed monodromy $ \rho = \chi(L)-\chi(0) $ by $H_{L}^{(\rho)} $.

Let $ \{ \eta_{n} \}_{n=1}^\infty $ be a sequence of normalized boundary height functions for the cylinder with circumference $ M_n $. We say the sequence is {\it stabilizing} if the normalized boundary height functions of $ \eta_{n} $ converge in the uniform metric to a function $ \eta \in H_L $.

Stabilizing Dirichlet boundary conditions\footnote{From now on we will refer to pure states corresponding to Dirichlet boundary conditions as to {\it boundary conditions} and to non-pure states as {\it boundary states}}  is a sequence of admissible boundary conditions $ \{ (\eta_n, \eta_n') \}_{n=1}^\infty $ stabilizing to $ ( \chi, \chi' ) $. The normalized free energy with stabilizing Dirichlet boundary conditions is:
\begin{align*}
f_{LT} [\chi,\chi']  =  -\lim_{n \rightarrow \infty} \epsilon_n^2 \log \left( Z_{M_n N_n}\big[\eta_{ n }, \; \eta_{ n }', \big] \right)
\end{align*}
Here $Z_{M_n N_n} \big[\eta_{n}, \eta_{n}' \big]$ is the partition function  with Dirichlet boundary conditions
$\eta_{n}, \eta_{n}'$ defined in (\ref{eq:partfunc}). We assume that the limit exists, but will not attempt to prove this.

We will say that the sequence of boundary states $b^{(n)}(\chi)$ is {\it stabilizing} as $n\to\infty$ if for each sequence of boundary height functions $\{\chi^{(n)}\}$ stabilizing to a piecewise smooth $\chi\in H_{L,\rho}$, there exists a function $ \beta $ on $[-1,1]\times [0,L]$ such that:
\[
b^{(n)}(\chi^{(n)})=\exp\left( -\frac{1}{\epsilon_n^2}\int_0^L\beta(\partial_x \chi(x),x) \; dx+o \left(\frac{1}{\epsilon_n^2}\right) \right)
\]
as $n\to \infty$. We assume that $ \beta $ is non-negative and smooth and we will also call such
states $\beta$-stabilizing. In the uniform distribution $\beta=0$.

We will be mostly interested in the situation when we assign a stabilizing boundary state to the upper end of the cylinder and a stabilizing boundary condition to the lower end.

Similarly to \cite{CKP}, one can argue that the limit $\epsilon_n\to 0$ of local correlation functions exists and
\[
\lim_{n\to \infty} \mathbb{E}_{M_nN_n}\Big[ \varphi(x_1,y_1) \varphi(x_2,y_2) \cdots \varphi(x_k,y_k) \Big] =\varphi_0(x_1,y_1)\cdots \varphi_0(x_k,y_k)
\]
here $\varphi$ is the random normalized height function on the cylinder $\mathcal{C}_{M_n N_n} $ and $\phi_0(x,y)$ is a function called the {\it limit shape}. Here we also assume appropriate stabilizing boundary conditions.

\subsubsection{The Variational Principle} Denote by $H_{LT}^{(\rho)} $ the space of functions $\varphi $ on $C_{LT} $ satisfying the Lipshitz conditions:
\begin{align*}
& | \varphi(x,y) - \varphi(x',y) | \leq |x -x'| \\
& | \varphi(x,y) - \varphi(x,y') | \leq |y-y'|
\end{align*}
We will call such functions macroscopic height functions.

Based on theorems for dimers in \cite{CKP}, it is natural to conjecture \cite{ZJ}\cite{PR}\cite{RS} that the
the system develops a limit shape, which is reflected in the behaviour of 1-point correlation functions discussed above and that the free energy for the cylinder with fixed stabilizing Dirichlet boundary conditions is given by the following variational
principle. Let $\chi,\chi'\in H_{L}^{(m)}$. Then
\[
f_{LT}[\chi,\chi']=-S[\varphi_0]
\]
where $\varphi_0$ is the minimizer of:
\begin{align}\label{eq:action}
S[\varphi] =&
\int_0^L \int_0^T \left(\sigma( \partial_x \varphi , \partial_y \varphi) + V \partial_x \varphi   + H \partial_y \varphi \right)\; dy \; dx
\end{align}
 in the space $H_{LT}^{(\rho)} $  with boundary conditions $\partial_x \varphi(x,T)=\partial_x \chi(x), \ \partial_x \varphi(x,0)=\partial_x \chi'(x) $. The function $\sigma(s,t)$ is the density of the free
energy for the 6-vertex model on a cylinder with given magnetizations (slopes) $(s,t)$ \cite{LW}. This function is the Legendre transform of the free energy on a cylinder as a function of magnetic fields.

In regions where $\nabla \varphi_0\in (1,-1)\times (1,-1)$ and avoids other singularities of $\sigma$, $\varphi_0$ satisfies the Euler-Lagrange equations:
\begin{equation}\label{ELeq}
\pa_x\sigma_1(\pa_x\varphi_0,\pa_y\varphi_0)+\pa_y\sigma_2(\pa_x\varphi_0,\pa_y\varphi_0)=0
\end{equation}
where $\sigma_1$ and $\sigma_2$ are partial derivatives in the first and the second
argument respectively.

The minimizer $\varphi_0$ is the limit shape, which means that it is the Gibbs measure of the 6-vertex model concentrates about $ \varphi_0 $ with variance exponentially supressed in $ \epsilon^2 $. It determines leading order of local correlation functions, as in the previous section.

For $\beta$-stabilizing boundary state on the upper end and the Dirichlet boundary condition $ \chi \in H_{L}^{(m)} $ on the lower end of the cylinder the free energy is given by\footnote{ The notation $f[\beta,\chi)$ indicates that at the upper boundary we have a $\beta$-stabilizing state and a the
boundary condition stabilizing at $\chi$ at the lower end.}
\[
f_{LT}[\chi,\beta)=\int_0^L\beta(\pa_x\varphi_0(x,T), x ) \; dx -S[\varphi_0]
\]
where $\varphi_0(x,y)$ is the minimizer of
\begin{equation}\label{svar}
S[\varphi]-\int_0^L\beta(\pa_x\varphi(x,T))\;dx
\end{equation}
 in the space $H_{LT}^{(m)} $  with boundary conditions $ \partial_x \phi(x,0)=\partial_x \chi(x) $.

The limit shape $\varphi_0$ satisfies the Euler-Lagrange equations:
\begin{equation}\label{ELeq}
\pa_x\sigma_1(\pa_x\varphi_0,\pa_y\varphi_0)+\pa_y\sigma_2(\pa_x\varphi_0,\pa_y\varphi_0)=0
\end{equation}
along with the boundary condition at $ y = T $:
\[
\sigma_1(\pa_x\varphi_0,\pa_y\varphi_0)+H+\pa_x( \beta_1(\pa_x\varphi_0,x) )=0
\]

\section{The Stochastic Six Vertex Weights and the ASEP Limit} \label{sec:stochastic}

\subsection{Stochastic Weights}
There is a natural discrete time stochastic process associated with the 6-vertex model, defined as follow. The state space $ \mathcal{S}_M^{(m)} $ of the process consist of boundary configurations of paths at the top end of the cylinder. The stochastic process $ \{ X_t \; | \; t \in \mathbb{N} \} $ is then defined by the conditional probabilities:
\begin{align}
\mathbb{P}[ X_t = s' \;  | \; X_0 = s] &= Z_{Mt}^{(m)}(s', s) / \mathcal{Z}_{Mt}[s] \label{eq:stochasticprocess} \\
\mathcal{Z}_t^{(m)}[s] &= \sum_{s'' \in \mathcal{S}} Z_{Mt}^{(m)}(s'',s) \label{eq:freebc}
\end{align}
\begin{remark}
The normalizing factor $ \mathcal{Z} $ is proportional to the partition function with Dirichlet boundary conditions on one end of the cylinder and the uniform distribution on the other.
\end{remark}

Recall that a discrete time stochastic process over a finite state space is \emph{Markov} if it satisfies $$ \mathbb{P}[ X_n = s_n \; | \; X_{n - 1} = s_{n-1}, \cdots, X_0 = s_0 ] = \mathbb{P} [X_1  = s_n \; |  \; X_{0} = s_{n-1}] $$ This means a Markov process is completely defined by "one step" transition probabilities $ \mathbb{P} [X_1  = x' \; |  \; X_{0} = x] = P_{x' x} $, which satisfy:
\begin{equation} \label{eq:markovmatrix}
\begin{aligned}
P_{x' x} \geq 0 \hspace{40pt} & \text{(positivity)} \\
 \sum_{x'} P_{x' x} = 1   \hspace{40pt} & \text{(total probability)}
\end{aligned}
\end{equation}
The total probability rule can be restated by saying that uniform vector $ \mathbbm{1} $  is a co-eigenvector (left eigenvector) of $ P $ with eigenvalue 1.
In other words, a transition matrix defines a Markov if it has nonnegative entries and the uniform distribution is a co-eigenvector of the matrix with the eigenvalue $1$.

\begin{proposition}\label{prop:stochasticweights}
With weights:
\begin{align} \label{eq:stochasticweights}
a_1 = 1 \;\;\; \; a_2 = 1  \;\;\; \; c_1 = 1-b_1 \;\;\; \;c_2 = 1-b_2
\end{align}
for $ b_1, b_2 \in [0,1] $, the process defined by (\ref{eq:stochasticprocess}) is Markov.
\end{proposition}
\noindent
\begin{proof}
Using lemma (\ref{lem:sum-trans}) from Appendix (\ref{sec:StoTrans}), we have for $ \mathcal{Z} $ as defined in (\ref{eq:freebc}):
\begin{align*}
\mathcal{Z}^{(m)}_1 (s) =  ( 1 + b_1^m \; b_2^{M-m} )
\end{align*}
 Note in particular that $ \mathcal{Z}^{(m)}_1(s) $ is independent of $ s $. It is clear from the definition (\ref{eq:stochasticprocess}) that the transition matrix for the process is the normalized transfer matrix
\begin{align*}
P^{(m)} = \frac{1}{1 + b_1^m b_2^{M-m} } T_M^{(m)}
\end{align*}
which satisfies the Markov conditions (\ref{eq:markovmatrix}).
\end{proof}

These values of Boltzmann weights are possible only if $\Delta=(a^2+b^2-c^2)/2ab>1$, i.e. the region where
in Baxter's parametrization $u,\eta>0$. In terms of this parametrization for the stochastic point we have:
\[
a_1=a_2=1, \ \ b_1=\frac{\sinh(u)e^{\pm\eta}}{\sinh(u+\eta)}, \ \ b_2=\frac{\sinh(u)e^{\mp\eta}}{\sinh(u+\eta)}
\]
\[
c_1=\frac{\sinh(\eta)e^{\mp u}}{\sinh(u+\eta)}, \ \ c_2=\frac{\sinh(\eta)e^{\pm u}}{\sinh(u+\eta)}
\]
For these values of weights $H=-V= \pm \eta / 2 $ and $\lambda=e^{\mp u}$.

From now on we will assume $b_2<b_1$ which correspond to $H=\eta>0$ and will call it the stochastic point. The other sign can be obtained by the reflection which exchanges the edges occupied by paths to empty edges. Note that the transfer matrices on the stochastic line form a commutative family.

The Markov property of the transfer matrix, together with the {\it uniform boundary condition at the upper end} and fixed boundary condition $\eta$ at the lower end implies that
\begin{itemize}
\item $\mathcal{Z}_{MN}[\eta]=1$ for any boundary state $\eta$.
\item The correlation functions of the 6-vertex model:
\begin{align*}
\mathbb{E}_{MN}\Big[ \varphi(x_1,y_1) \varphi(x_2,y_2) \cdots \varphi(x_k,y_k) \Big]
\end{align*}
are equal for all $ N \geq \max_i y $.
\end{itemize}

In the thermodynamic limit, the limit shape phenomenon in this case states
that the random variable $\varphi(x,y) $ on the cylinder $ C_{LT} $ converges in probability to a limit shape $\varphi_0(x,y)$. As a consequence the correlation function written above is expected to converge to
\[
\varphi_0(x_1,y_1)\cdots \varphi_0(x_k,y_k)
\]
where $\varphi_0$ is the solution to the variational problem \ref{svar}. We will discuss $\varphi_0$ in the
next sections.

The limit shape $\varphi_0(x,y)$ with this boundary condition is the most probable height
function in the bulk and also the most probable boundary height function at the top boundary.
The limit shape phenomenon means that random height functions concentrate around $\varphi_0$.
The limit shape $\varphi_0$ is completely determined by the boundary conditions at the lower end of the
cylinder. This suggests that it should be determined by a first order PDE (which we will
derive in the next section) and that $\varphi_0$
is the restriction of a solution to this PDE defined for the half infinite cylinder $[0,L] \times \mathbb{R}^+ \rightarrow \mathbb{R}$ to $[0,L]\times [0,T]$.

\subsection{The ASEP Limit} Here we will show that transition probabilities of the \emph{asymmetric exclusion process} (ASEP) can be obtained
from the commuting family of stochastic 6-vertex transfer matrices in a natural limit near the point $u=0$.

\subsubsection{The Asymmetric Exclusion Process on a Ring}
Recall that a continuous time Markov process $ X(t) $ over a finite state space is defined by a transition matrix $ W $ as $
dp_i / dt = \sum_{j=1}^M W_{ij} p_j $, where $ p_i(t) = \mathbb{P}[X(t) = i], i=1,\cdots, M $ is the probability that system is in the $ i $th state at time $ t $. The transition matrix satisfies:
\begin{equation}
\begin{aligned}
W_{ij} > 0 \; \;  \text{for $ i\neq j$ }  \hspace{40pt} & \text{(positivity)} \\
 \sum_{j=1}^M W_{ij} = 0  \hspace{40pt} & \text{(total probability)}
\end{aligned}
\end{equation}
The finite time transition probabilities $ P_{ij}(t) = \mathbb{P}[X(t) = i \; | \; X(0) = j ] $ are determined by the differential equation:
\begin{align*}
\frac{d}{dt} P_{ij}(t) &= \sum_{k=1}^M W_{ik} P_{kj}(t)  \\
P_{ij}(0) &= \delta_{i,j}
\end{align*}
For any fixed $ t > 0 $, they satisfy (\ref{eq:markovmatrix}) and define Markov transition probablilities.
Note that such process is invariant with respect to transformations $t\mapsto t\lambda^{-1}$, $W_{ij}\mapsto W_{ij}\lambda$ for positive $\lambda$. Thus only projective values of $W_{ij}$ are important.

The  ASEP is a continuous time stochastic process describing an interacting particle system. The states of the process are configurations of particles on the ring $ \mathbb{Z} / M \mathbb{Z} $ with at most one particle at each lattice site. Each particle waits an exponentially distributed random time, then chooses the  adjacent lattice site to the left with probability $ p$, or the right with probability $ 1- p $. If the chosen site is empty, the particle jumps over; otherwise it remains at the same site.

The transition matrix for ASEP can be written as:
\begin{align*}
W_{ASEP} = \sum_{i}^{M} H_{i,i+1}(p,q)
\end{align*}
where
\begin{align*}
H(p,q) = \begin{pmatrix}
0 & 0 & 0 & 0 \\
0 & q & -p & 0 \\
0 & -q & p & 0 \\
0 & 0 & 0 & 0
\end{pmatrix}
\end{align*}
$ H_{i,j}(p,q) $ acts as $ H $ on the $ i$th and $ j$th components of the tensor product $ (\mathbb{C}^2)^{\otimes M} $, and $  H_{M,M+1} $ is defined to be $ H_{M,1} $ (periodic boundary conditions).

\subsubsection{The Relation Between ASEP and the Stochastic 6-vertex Model}\label{sec:ASEPLim}
The following is a well known relation between the 6-vertex model and the Heisenberg spin chain.
\begin{proposition} The matrix of transition probabilities for the ASEP model is related to
the transfer matrix of the stochastic 6-vertex models as follows:
\[
W=\sum_{i=1}^MH_{i,i+1}(p,q)=T'_M(0, H,V)T_M(0,H,V)^{-1}-\text{coth}(\eta)I
\]
where $p=\frac{e^{\eta}}{\sinh(\eta)}, \ \ q=\frac{e^{-\eta}}{\sinh(\eta)}$.
Note that after rescaling this covers all possible values $0<q<p<1$. Here $ T'(u,H,V) = \frac{d}{du} T(u,H,V) $.
\end{proposition}
\begin{proof}
First note that as $u\to 0$ we have
\[
R(u)=\sinh(\eta)P(1+u\tilde{H}+O(u^2))
\]
where $ P $ is the permutation operator, $ P x \otimes y = y \otimes x $ and
\begin{align*}
\tilde{H}=\begin{pmatrix}
\text{coth}(\eta)& 0 & 0 & 0 \\
0 & 0 & \text{csch}(\eta) & 0 \\
0 & \text{csch}(\eta) & 0 & 0 \\
0 & 0 & 0 & \text{coth}(\eta)
\end{pmatrix}
\end{align*}

From here
\[
T_M(0,H,V)=\sinh(\eta)^M (D^{2H+2V}_1\cdots D^{H+V}_M) C
\]
where $C$ is the cyclic permutation
\begin{align} \label{eq:cyclic}
C(x_1\otimes \cdots \otimes x_M)=x_2\otimes \cdots \otimes x_N\otimes x_1
\end{align}  and $D^H$ is the same diagonal matrix as before. For the derivative in $u$ at $u=0$ we have:
\[
T'_M(0, H,V)=\sinh(\eta)^M (D^{2H+2V}_1\cdots D^{H+V}_M)C\sum_{i=1}^M U_{i,i+1}
\]
where
\[
U_{a,b}=D_b^{-H}\tilde{H}_{ba}D_b^{H}
\]
This operator can be written as
\[
U=\text{coth}(\eta)I+\sinh(\eta)(I\otimes A-A\otimes I)+H(p,q)
\]
where
\begin{align*}
p=\frac{e^{\eta}}{\sinh(\eta)}, \ \ q=\frac{e^{-\eta}}{\sinh(\eta)}, \ \ A=\begin{pmatrix}0&0\\0&1\end{pmatrix}
\end{align*}
Taking into account that $\sum_{i=1}^M (A_i-A_{i+1})=0$ we arrive to
\[
W=T'_M(0, H,V)T_M(0,H,V)^{-1}-\text{coth}(\eta) \; I
\]
\end{proof}

\subsubsection{ASEP Height Functions and the Thermodynamical Limit}
To an ASEP state $ S $, we associate a continuous real valued function on $ \mathbb{R} $ as follows. The function $ h_S $ is linear on each interval $ [i-1/2, i+1/2] $ for $ i \in \mathbb{N} $, with slope $ 1 $ if the $ i$th lattice site is occupied by a particle, and slope $ -1 $ otherwise.

\begin{figure}[h]
\begin{tikzpicture}[scale= 1]
\coordinate (e0) at (1.5, 0);
\coordinate (e1) at (0.935235, 1.17275);
\coordinate (e2) at (-0.333781, 1.46239);
\coordinate (e3) at (-1.35145, 0.650826);
\coordinate (e4) at (-1.35145, -0.650826);
\coordinate (e5) at (-0.333781, -1.46239);
\coordinate (e6) at (0.935235, -1.17275);

\draw (0,0) circle (1.5);

\fill[white] (e0) circle (2 pt);
\fill[white] (e1) circle (2 pt);
\fill[white] (e2) circle (2 pt);
\fill[white] (e3) circle (2 pt);
\fill[white] (e4) circle (2 pt);
\fill[white] (e5) circle (2 pt);
\fill[white] (e6) circle (2 pt);

\draw (e0) circle (2 pt);
\draw (e1) circle (2 pt);
\draw (e2) circle (2 pt);
\draw (e3) circle (2 pt);
\draw (e4) circle (2 pt);
\draw (e5) circle (2 pt);
\draw (e6) circle (2 pt);

\fill (e1) circle (2 pt);
\fill (e4) circle (2 pt);
\fill (e5) circle (2 pt);

\end{tikzpicture} \; \; \; \;
\begin{tikzpicture}[scale = 1,baseline]

\draw (0,0)--(6,0);
\coordinate (e0) at (0,0);
\coordinate (e1) at (1,0);
\coordinate (e2) at (2,0);
\coordinate (e3) at (3,0);
\coordinate (e4) at (4,0);
\coordinate (e5) at (5,0);
\coordinate (e6) at (6,0);

\fill[white] (e0) circle (2 pt);
\fill[white] (e1) circle (2 pt);
\fill[white] (e2) circle (2 pt);
\fill[white] (e3) circle (2 pt);
\fill[white] (e4) circle (2 pt);
\fill[white] (e5) circle (2 pt);
\fill[white] (e6) circle (2 pt);

\draw (e0) circle (2 pt);
\draw (e1) circle (2 pt);
\draw (e2) circle (2 pt);
\draw (e3) circle (2 pt);
\draw (e4) circle (2 pt);
\draw (e5) circle (2 pt);
\draw (e6) circle (2 pt);

\fill (e1) circle (2 pt);
\fill (e4) circle (2 pt);
\fill (e5) circle (2 pt);

\draw (0,2)--(.5,2-.5)--(1.5,2.5)--(2.5,2-.5)--(3.5,2-1.5)--(4.5,2-.5)--(5.5,2.5)--(6,2);

\end{tikzpicture}
\caption{An ASEP state and the height function.} \label{fig:asep}
\end{figure}
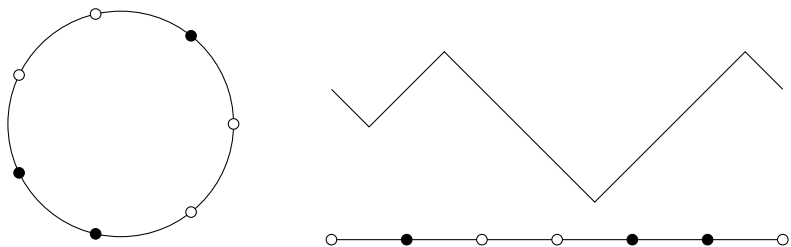 

As in the stochastic 6-vertex model one can introduce the normalized ASEP height function on $ [0,M] $  as $ \phi(x) = \epsilon  h(\lfloor x / \epsilon \rfloor) $, where the mesh $ \epsilon = L/M $. The limit shape phenomenon for the
normalized height function of the ASEP model means that as $\epsilon \to 0$, the normalized random height function $\phi(x,t) = \epsilon  h( \lfloor x / \epsilon \rfloor, t)$ converges in probability to a deterministic function $\phi_0(x,t)$. In particular, we have:
\[
\lim_{\epsilon \to 0} \mathbb{E}_{ \lfloor \frac{N}{\epsilon} \rfloor }\left(\phi(x_1,t_1) \phi(x_2,t_2) \cdots \phi(x_k,t_k) \right)=\phi_0(x_1,t_1)\cdots \phi_0(x_k,t_k)
\]
The function $\phi_0(x,t)$ is called the limit shape. It is been known the limit shape is the solution to the
inviscid Burgers equation \cite{BF} \cite{Rez}. We will see that this also can be obtained as a limit of the
limit shape equation for the stochastic 6-vertex model.

\section{Limit Shapes for the stochastic 6-vertex model}

In this section we will study limit shapes for the stochastic 6-vertex model with the
uniform distribution at the upper end ($y=T$) and Dirichlet boundary conditions on the
the lower end ($y=0$).  The main result of this section is the partial differential equations for the limit shape of the stochastic 6-vertex model.

\subsubsection{The Variational Problem} Consider the variation of the action functional with free boundary condition at the upper end of the cylinder and Dirichlet boundary condition at the lower end.
In what follows $H_0=-V_0=\eta/2$.
\begin{align}\label{eff-action}
\delta S[\varphi ] =&
\int_0^L\int_0^T \Big(\partial_1\sigma( \varphi_x , \varphi_y  )\delta\varphi_x+ \partial_2\sigma( \varphi_x , \varphi_y  )\delta\varphi_y +   \; V_0  \; \delta \varphi_x  +  \; H_0 \; \delta \varphi_y\Big)  \; dx \; dy \\
= & - \int_0^L\int_0^T  \Big(\partial_x \sigma_1( \varphi_x, \varphi_y) + \partial_y \sigma_2( \varphi_x, \varphi_y)\Big) \; \delta \varphi \;  dx \;dy - \int_{0}^L \big(\sigma_2(\varphi_x, \varphi_y) -  \; H_0 \big) \; \delta \varphi \; dx
\end{align}
Here we assumed that $\varphi(x,0)=\chi(x)$ and therefore  $\delta\phi(x,0)=0$. When $\nabla \varphi$ does not hit a singularity of $\sigma(s,t)$ the first term gives the Euler-Lagrange equation:
\begin{align}
\partial_x \sigma_1( \varphi_x, \varphi_y) + \partial_y \sigma_2( \varphi_x, \varphi_y) = 0 \; \; \; \;
\end{align} \label{eq:eom}
The second term gives boundary constraint at $y=T$:
\[
\sigma_2\big(\varphi_x(x,T) , \varphi_y(x,T) \big) -H_0=0 \; \; \; \;
\]

\subsubsection{Degeneration to a Parabolic PDE}
A crucial property of the function $ \sigma $ described in Appendix \ref{sec:Surf} is:

\begin{lemma} \label{zero-He} The Hessian of $\sigma$\footnote{ $\text{Hess}\left( \sigma(s, t) \right) = \sigma_{11}(s, t)\sigma_{22}(s, t)-\sigma_{12}(s, t)^2 $} when $\Delta>1$ is vanishing on the curve $ \sigma_2(s,t) = \pm H_0 $.
\end{lemma}

\noindent As a consequence, we have
\begin{lemma} \label{lem:peq} The constant:
\begin{align*}
\sigma_2\big(\varphi_x(x,y) , \varphi_y(x,y) \big) = H_0 \; \; \; \; & \text{(for all $ (x,y) \in C_{LT} $)}
\end{align*}
is consistent with the Euler-Lagrange equations. This equation describes the minimizer of the
functional $S[\varphi]$ which
satisfies the free boundary condition at any given $y$.
\end{lemma}
\begin{proof}
Starting with the bulk equation of motion (\ref{eq:eom}):
\begin{align*}
\partial_y \sigma_2 &= -\partial_x \sigma_1 \\
&= -\sigma_{11} \varphi_{xx} - \sigma_{12} \varphi_{xy} \\
&= - \frac{\sigma_{12}^2}{\sigma_{22}} \varphi_{xx} - \sigma_{12} \varphi_{xy}
\end{align*}
Here we used $ \sigma_{11} \sigma_{22} - \sigma_{12}^2 = 0 $ to eliminate $ \sigma_{11} $.
\begin{align*}
&= -\frac{\sigma_{12}}{\sigma_{22} } \left( \sigma_{12} \varphi_{xx} + \sigma_{22} \varphi_{xy} \right) \\
&=  -\frac{\sigma_{12}}{\sigma_{22} } \left( \partial_x \sigma_2 \right) \\
&= 0
\end{align*}
\end{proof}

\begin{remark}
The fact that the uniform boundary condition gives the constraint $ \sigma_2\big(\varphi_x , \varphi_y \big) = H_0 $ that is preserved by the equations of motion is analagous (and is related) to the fact that the uniform boundary condition is a coeigenvector of the Markov process with maximal eigenvalue.
\end{remark}

Using the lemma, the first term in the bulk equation (\ref{eq:eom}) vanishes, and the limit shape equation becomes parabolic:
\begin{align*}
0 &=\partial_x \sigma_1(\varphi_x, \varphi_y)  \\
&= \sigma_{11}(\varphi_x, \varphi_y) \; \varphi_{xx} +\sigma_{12}( \varphi_x,\varphi_y) \; \varphi_{xy}
\end{align*}
The partial derivatives of the surface tension can be computed exactly on the critical line \cite{BS}. Alternatively, we can exploit the exact formula for the shape of the line (\ref{eq:critline}):
\begin{proposition}
Limit Shape Equation:
\begin{align} \label{eq:limshape}
\rho_y = - \frac{v^2 -1}{( 1 + v \rho )^2 } \; \rho_x
\end{align}
where $ \rho = \varphi_x $ and:
\begin{align*}
v = \frac{b_1- b_2}{2 - b_1 - b_2}
\end{align*}
Note that in Baxter's parametrization $v=\tanh(u+\eta)$.
\end{proposition}
\begin{proof}
Since $ \sigma_2(\varphi_x, \varphi_y) = H_0 $, the slopes must lie they must lie on the critical line (\ref{eq:critline}):
\begin{align}
\varphi_y = \frac{\varphi_x + \tanh(u+\eta)}{1 + \tanh(u+\eta) \varphi_x }
\end{align}

Differentiating with respect to $ x $ gives:
\begin{align*}
\varphi_{xy} &= \frac{1 - \tanh(u+\eta)^2}{ (1 + \tanh(u+\eta) \varphi_{x})^2 } \; \varphi_{xx}
\end{align*}
Substituting $ \rho = \varphi_x $  gives (\ref{eq:limshape}).
\end{proof}

Note that we are interested in periodic solutions to (\ref{eq:limshape}) with $L$-periodic initial conditions.

\subsection{The ASEP Limit and Burgers Equations}
It is well known \cite{BF}\cite{Rez}\cite{GS} that in the hydrodynamic limit of ASEP, the time evolution of the ASEP height function is described by the inviscid Burgers' equation. This agrees with our result after taking in (\ref{eq:limshape}) the ASEP limit as described in section (\ref{sec:ASEPLim}).

Letting $ b_1 = \epsilon p $ and $ b_2 = \epsilon q $, we expand (\ref{eq:limshape}) in $\epsilon $:
\begin{align*}
\rho_t = \rho_x - \frac{1}{2} \epsilon (p-q) \rho \rho_x + O(\epsilon^2)
\end{align*}
The limit $\epsilon\to 0$ corresponds to $u\to 0$.
The first term $ \rho_t = \rho_x $ is the transport equation and is the continuum limit of the cyclic shift operator (\ref{eq:cyclic}).

They can be removed with a change of coordinates to $ (\w x, \w y) $:
\begin{align*}
x &= \w x & y &= \epsilon \w y + x \\
\partial_x &= \partial_{\w x} & \partial_y &= \epsilon \partial_{\w y} + \partial_{\w x}
\end{align*}
Then in the limit $ \epsilon \rightarrow 0 $:
\begin{align*}
\rho_{\w y} =\frac{1}{2} (p-q) \; \rho \; \rho_{\w x}
\end{align*}

Note that in Baxter's parametrization $\frac{1}{2}(p-q)=1$
\section{Solutions to the Limit Shape equations}
Recall $ \rho = \partial_x \varphi $ satisfies $ \rho \in [-1, 1] $. It is useful to write the limit shape PDE in conservation form:
\begin{align} \label{eq:cons}
 \rho_y &= \partial_x \left( F(\rho) \right) \\
&=  f(\rho) \;  \rho_x
\end{align}
where $ f $ is as given in (\ref{eq:limshape}), and $ F $ is an antiderivative.

Note that when solutions are smooth and $ F $ is strictly convex, the PDE can be transformed into the inviscid Burgers' equation, letting:
\begin{align} \label{eq:invburg}
u(x,y) = f(\rho(x,y) )
\end{align}
Then $ \rho $ satisfies (\ref{eq:cons}) if and only if $ u$ satisfies $u_y = u u_x $.

\subsection{Characteristics for Conservation Equations}
Equations of form (\ref{eq:cons}) can be solved by the method of characteristics. Assume the initial condition $  \rho(x,0) = \rho_0(x) $. Then, assuming the solution is smooth, the solution $ \rho(x,y) $ is constant on the line in spacetime that passes through $ (x_0,0) $ with slope $ dx/dy = f( \rho_0(x_0) ) $, i.e. on the line:
\begin{align*}
x = f \left( \rho(x_0) \right) y + x_0
\end{align*}

\begin{figure}[H]
\begin{tikzpicture}[ scale = 1.5]
\draw[ultra thin, gray, decoration={markings,mark=at position 1 with {\arrow[scale=2]{>}}},
    postaction={decorate}](0,-.8) -- (0,.9);
\draw[ ultra thin, gray, decoration={markings,mark=at position 1 with {\arrow[scale=2]{>}}},
    postaction={decorate}](-1.4,0) -- (1.4,0);

\node at  (1.6,0) { $ x $ };
\node at  (0,1.1) { $ \rho_0 $ };

\draw [domain = -1.3:1.3, smooth] plot ({\x},{1/2-1/(1+exp(3*\x))});
\end{tikzpicture} \; \; \; \; \; \; \;
\begin{tikzpicture}[ scale = 1.5]
\draw[ultra thin, gray, decoration={markings,mark=at position 1 with {\arrow[scale=2]{>}}},
    postaction={decorate}](0,0) -- (0,1.7);
\draw[ ultra thin, gray, decoration={markings,mark=at position 1 with {\arrow[scale=2]{>}}},
    postaction={decorate}](-1.4,0) -- (1.4,0);

\node at  (1.6,0) { $ x $ };
\node at  (0,1.9) { $ y $ };

\begin{scope}
\clip (-1.3,0) rectangle (1.3,1.5);

\draw (.3,0)-- (.3+.2,2);
\draw (.6,0)-- (.6+.5,2);
\draw (.9,0)-- (.9+.7,2);
\draw (1.2,0)--(1.2+.75,2);

\draw (0,0)--(0,2);

\draw (-.3,0)--(-.3-.2,2);
\draw (-.6,0)--(-.6-.5, 2);
\draw (-.9,0)--(-.9-.7,2);
\draw (-1.2,0)--(-1.2-.75,2);
\end{scope}
\end{tikzpicture}
 
\caption{An initial condition $ \rho_0 $ and its characteristic lines with $ f(\rho) = \rho $ } \label{fig:shock}
\end{figure}
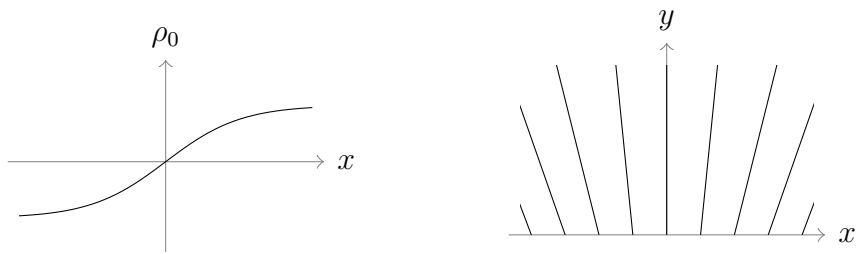

Solutions to the Burgers equations with generic initial conditions typically develop discontinuities at later times. These discontinuities correspond to regions where the characteristics fail, either giving multiple solutions or no solutions. We instead look for weak solutions to the PDE \cite{Evans}.

\subsubsection{Shocks} Intersecting characteristic lines give multiple values $ \rho $ for each point in space time. Thee correspond to a jump discontinuity in $ \rho $ called a shock. The precise location of the shock within a region of intersecting characteristics  is determined by the \emph{Rankine-Hugoniot} condition for the propogation of the shock. If $ s(t) $ is the $x$ coordinate of shock then:
\begin{align*}
\frac{ds}{dy}(y) =\frac{ F(\rho_L(y)) - F(\rho_R(y))}{\rho_L(y) - \rho_R(y)} \hspace{30pt}   \rho_{L,R}(y) = \lim_{x \rightarrow s(y)^{\mp} } \rho(x,y)
\end{align*}
Here $ \rho_L $ and $ \rho_R $ are the values of $ \rho $ just to the left and right side of the jump discontinuity. Note that the jump condition and propogation of the shock is not preserved by the transformation (\ref{eq:invburg}).

\begin{figure}[h]
\begin{tikzpicture}[ scale = 1.5]
\draw[ultra thin, gray, decoration={markings,mark=at position 1 with {\arrow[scale=2]{>}}},
    postaction={decorate}](0,0) -- (0,1.7);
\draw[ ultra thin, gray, decoration={markings,mark=at position 1 with {\arrow[scale=2]{>}}},
    postaction={decorate}](-1.4,0) -- (1.4,0);

\node at  (1.6,0) { $ x $ };
\node at  (0,1.9) { $  y $ };

\draw [domain = -1:-.1, smooth, ultra thick] plot ({\x-.05},{1.75*sqrt(-\x)});

\draw (0,0) -- (-.1-.05,0.553399);
\draw (.3,0)--(-.1-.05,0.553399);
\draw (.6,0)--(-.3-.05,0.958515);
\draw (.9,0)--(-.6-.05,1.35555);
\draw (1.2,0)--(-.90,1.61342);

\draw (-.3,0)--(-.3-.05,0.958515);
\draw (-.6,0)--(-.6-.05,1.35555);
\draw (-.9,0)--(-.90,1.61342);
\end{tikzpicture}
 
\caption{A shock (drawn in bold) caused by the intersecting characteristics.} \label{fig:shock}
\end{figure}
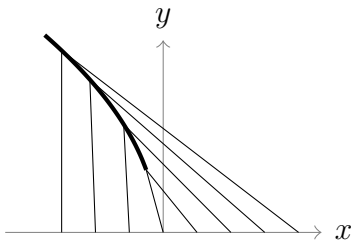

\subsubsection{Rarefaction Fans} Diverging characteristics lead to regions of spacetime where $ \rho $ is not determined by the characteristics. The correct solution for $ \rho $ in these regions, called rarefaction fans, is given by the \emph{Lax-Oleinik} solution:
\begin{align*}
\rho(x,y) = f^{-1}\left( \frac{x}{y} \right) \hspace{30pt} f( \rho_L) < \frac{x}{y} < f(\rho_R)
\end{align*}
Here $ \rho_+ $ and $ \rho_- $ are the values of $ \rho $ just beyond the boundary of the rarefaction fan; in other words $ F( \rho_L) $ and $ F( \rho_R) $ are the slopes of the characteristics bounding the fan.
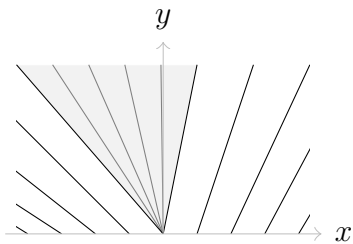
\begin{figure}[h]
\begin{tikzpicture}[ scale = 1.5]

\node at  (1.6,0) { $ x $ };
\node at  (0,1.9) { $ y $ };

\begin{scope}

\draw [white, fill = black!5!white] (0,0) -- (-1.3,1.5) --  ( 0.3,1.5) -- (0,0);

\clip (-1.3,0) rectangle (1.3,1.5);

\draw (-.3,0)-- (-.3-1.5,1.5);
\draw (-.6,0)--  (-.6-1.9,1.5);
\draw (-.9,0)--  (-.9-2.3,1.5);
\draw (-1.2,0)-- (-.9-2.7,1.5);

\draw [gray] (0,0) -- (-1.3+.32,1.5);
\draw [gray] (0,0) -- (-1.3+2*.32,1.5);
\draw [gray] (0,0) -- (-1.3+3*.32,1.5);
\draw [gray] (0,0) -- (-1.3+4*.32,1.5);

\draw (0,0) -- (-1.3,1.5);
\draw (0,0) -- ( 0.3,1.5);

\draw (.3,0)--(0.3+.5,1.5);
\draw (.6,0)-- (.6+.7,1.5);
\draw (.9,0)--(.9+.8, 1.5);
\draw (1.2,0)--(1.2+.9, 1.5);

\end{scope}

\draw[ultra thin, lightgray, decoration={markings,mark=at position 1 with {\arrow[scale=2]{>}}},
    postaction={decorate}](0,0) -- (0,1.7);
\draw[ ultra thin, lightgray, decoration={markings,mark=at position 1 with {\arrow[scale=2]{>}}},
    postaction={decorate}](-1.4,0) -- (1.4,0);

\end{tikzpicture}
 
\caption{A rarefaction fan and its characteristics in gray.} \label{fig:shock}
\end{figure}

\subsection{Limit Shapes of the Stochastic Six Vertex Model with Domain Wall Initial Conditions} Recall that for the stochastic 6-vertex model:
\begin{align*}
 F(\rho) = \frac{1}{v} \; \frac{1-v^2}{1 + \rho \; v} \hspace{30pt} f(\rho) = F'(\rho) = -\frac{1-v^2}{(1 + v \rho)^2}
\end{align*}
where
\begin{align*}
 v = \frac{b_2 - b_1}{2 - b_2 - b_1}
\end{align*}
Without loss of generality we assume $ 0 < b_1 < b_2 < 1 $ so that $ v \in (0,1) $.  Then $ F $ is strictly convex since $ F''(\rho) = f'(\rho) > 0 $, see Figure \ref{fig:Ff}.

The critical values of of $ \rho = \pm 1 $ correspond to "frozen" phases of the 6-vertex model.  For these:
\begin{align*}
f(-1) &= \frac{v+1}{v-1} = \frac{b_2 - 1}{b_1 - 1} \\
f(1) &= \frac{v-1}{v+1} = \frac{b_1 - 1}{b_2 - 1}
\end{align*}

The inverse of $ f $ is easily calculated:
\begin{align} \label{eq:fi}
f^{-1}(x) = \frac{-1 +  \sqrt{\frac{1-v^2}{|x| } }}{v}
\end{align}
for $ f(-1) < x < f(1) < 0 $.  For reference we also calculate:
\begin{align*}
F(f^{-1}(x) ) = \sqrt{\left( 1- \frac{1}{v^2} \right) |x| }
\end{align*}

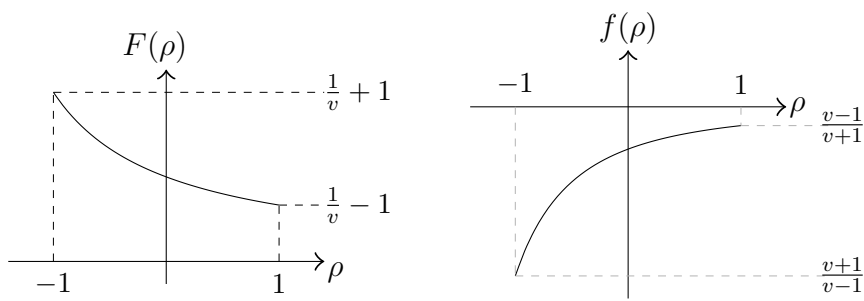
\begin{figure}[H]
\begin{tikzpicture}[ scale = 1.5]
\draw[thin, decoration={markings,mark=at position 1 with {\arrow[scale=2]{>}}},
    postaction={decorate}](0,-.2) -- (0,1.7);
\draw[ thin, decoration={markings,mark=at position 1 with {\arrow[scale=2]{>}}},
    postaction={decorate}](-1.4,0) -- (1.4,0);

\node at  (1.5,-.1) { $ \rho $ };
\node at  (-.1,1.9) { $ F(\rho) $ };

\draw [domain = -1:1, smooth] plot ({\x},{.75/( 1 + \x*.5)});

\draw [ultra thin, dashed] (-1,1.5)--(1.4,1.5);
\draw [ultra thin, dashed] (1,.5)--(1.4,.5);
\node at (1.7,.5) { \small{$ \frac{1}{v} -1 $} };
\node at (1.7,1.5) { \small{$ \frac{1}{v} +1 $} };

\draw [ultra thin, dashed] (-1,0)--(-1,1.5);
\draw [ultra thin, dashed] (1,0)--(1,.5);
\node at (1,-.20) { \small{$ 1 $} };
\node at (-1,-.20) { \small{$ -1 $} };

\end{tikzpicture} \; \; \;
\begin{tikzpicture}[scale = 1.5]
\draw[thin, decoration={markings,mark=at position 1 with {\arrow[scale=2]{>}}},
    postaction={decorate}](0,-1.7) -- (0,.5);
\draw[ thin, decoration={markings,mark=at position 1 with {\arrow[scale=2]{>}}},
    postaction={decorate}](-1.4,0) -- (1.4,0);

\node at  (1.5,0) { $ \rho $ };
\node at  (0,.7) { $ f(\rho) $ };

\draw [domain = -1:1, smooth] plot ({\x},{-.375/(( 1 + \x*.5)*( 1 + \x*.5))});

\draw [ultra thin, black!30!white, dashed] (-1,-1.5)--(1.8,-1.5); \node at (1.9,-1.5) { \small{$ \frac{v+1}{v-1} $} };
\draw [ultra thin, black!30!white, dashed] (1,-.16666)--(1.8,-.16666); \node at (1.9,-.166666) { \small{$ \frac{v-1}{v+1} $} };

\draw [ultra thin, black!30!white, dashed] (-1,0)--(-1,-1.5);
\draw [ultra thin, black!30!white, dashed] (1,0)--(1,-.1666666);
\node at (1,.20) { \small{$ 1 $} };
\node at (-1,.20) { \small{$ -1 $} };

\end{tikzpicture}
 
\caption{Plots of $ F(\rho) $ and $ f(\rho) $.} \label{fig:Ff}
\end{figure}

\subsubsection{Example 1}
Consider the step initial conditions on a line ($L=\infty, x\in \mathbb{R}$)
\begin{align*}
\rho_0  =  \begin{cases} \rho_R \; & x \geq 0 \\ \rho_L \; & x < 0 \end{cases}
\end{align*}
When $\rho_L>\rho_R$ a shock forms and  propagates from the origin with constant velocity:
\begin{align*}
\frac{ds}{dy} = \frac{F(\rho_L) - F(\rho_R)}{\rho_L-\rho_R} =v_0
\end{align*}
The corresponding solution is:
\begin{align*}
\rho(x,y) = \begin{cases} \rho_R \; & x/y \geq v_0 \\ \rho_L \; & x/y < v_0 \end{cases}
\end{align*}
\begin{figure}[H]
\begin{tikzpicture}[ scale = 1.5]
\draw[ultra thin, gray, decoration={markings,mark=at position 1 with {\arrow[scale=2]{>}}},
    postaction={decorate}](0,-.7) -- (0,.7);
\draw[ ultra thin, gray, decoration={markings,mark=at position 1 with {\arrow[scale=2]{>}}},
    postaction={decorate}](-1.4,0) -- (1.4,0);

\node at  (1.55,0) { $ x $ };
\node at  (0,.85) { $ \rho_0 $ };

\draw (1.2,-.5)--(0,-.5);
\draw (0,.5)--(-1.2,.5);
\end{tikzpicture}  \; \; \; \; \; \; \;
\begin{tikzpicture}[ scale = 1.5]
\draw[ultra thin, gray, decoration={markings,mark=at position 1 with {\arrow[scale=2]{>}}},
    postaction={decorate}](0,0) -- (0,1.7);
\draw[ ultra thin, gray, decoration={markings,mark=at position 1 with {\arrow[scale=2]{>}}},
    postaction={decorate}](-1.4,0) -- (1.4,0);

\node at  (1.6,0) { $ x $ };
\node at  (0,1.9) { $ y $ };

\draw [ultra thick] (0,0)--(-0.824176,1.5);
\draw (.3,0)-- (-0.3500,0.637002);
\draw (.6,0)-- (-0.70000,1.274);
\draw (.9,0)-- (.9-1.53061,1.5);
\draw (1.2,0)--(1.2-1.53061,1.5);

\draw (-.3,0)--(-0.3500,0.637002);
\draw (-.6,0)-- (-0.70000,1.274);
\draw (-.9,0)--(-.9-0.117739, 1.5);
\draw (-1.2,0)--(-1.2-0.117739, 1.5);
\end{tikzpicture}
\end{figure}
\noindent When $\rho_R=-1$ and $\rho_L=1$ we have $v_0=-1$.
\subsubsection{Example 2}
Consider now the step initial conditions on a line with $\rho_L<\rho_R$. In this case
characteristic solutions do not span the region $f(\rho_L)<x/y<f(\rho_R)$ and instead
we have a rarefaction fan propagating in this region from the origin. Using the Lax-Oleinik formula and (\ref{eq:fi}), we have the general solution:
\begin{align*}
\rho(x,y) = \begin{cases} \rho_R \; & x/y > f(\rho_R) \\  \frac{-1 +  \sqrt{\frac{1-v^2}{x/y} }}{v}  \; & f(\rho_L) < x/y < f(\rho_R)  \\ f(\rho_L) \; & x/y < f(\rho_L) \end{cases}
\end{align*}
For critical densities $\rho_L=-1$, $\rho_R=1$ we have
\[
f(\rho_L)=\frac{v-1}{v+1}, \ \ f(\rho_R)=\frac{v+1}{v-1}
\]
and this solution correspond to the limit shape derived in \cite{BCG}.
\begin{figure}[H]
\begin{tikzpicture}[ scale = 1.5]
\draw[ultra thin, gray, decoration={markings,mark=at position 1 with {\arrow[scale=2]{>}}},
    postaction={decorate}](0,-.7) -- (0,.7);
\draw[ ultra thin, gray, decoration={markings,mark=at position 1 with {\arrow[scale=2]{>}}},
    postaction={decorate}](-1.4,0) -- (1.4,0);

\node at  (1.55,0) { $ x $ };
\node at  (0,.85) { $ \rho_0 $ };

\draw (1.2,.5)--(0,.5);
\draw (0,-.5)--(-1.2,-.5);
\end{tikzpicture} \; \; \; \; \; \; \;
\begin{tikzpicture}[ scale = 1.5]
\draw[ultra thin, gray, decoration={markings,mark=at position 1 with {\arrow[scale=2]{>}}},
    postaction={decorate}](0,0) -- (0,1.7);
\draw[ ultra thin, gray, decoration={markings,mark=at position 1 with {\arrow[scale=2]{>}}},
    postaction={decorate}](-1.4,0) -- (1.4,0);

\node at  (1.6,0) { $ x $ };
\node at  (0,1.9) { $ y $ };

\draw (-.3,0)-- (-1.4,1.078);
\draw (-.6,0)-- (-1.4,0.784001);
\draw (-.9,0)-- (-1.4,0.490001);
\draw (-1.2,0)--(-1.4,0.196);

\draw [ultra thin] (0,0) -- (-1.2,1.5);
\draw [ultra thin] (0,0) -- (-.9,1.5);
\draw [ultra thin] (0,0) -- (-.6,1.5);
\draw [ultra thin] (0,0) -- (-.3,1.5);

\draw (0,0) -- (-1.4,1.372);
\draw (0,0) -- ( 0-0.117739,1.5);

\draw (.3,0)--(0.3-0.117739,1.5);
\draw (.6,0)-- (0.6-0.117739,1.5);
\draw (.9,0)--(.9-0.117739, 1.5);
\draw (1.2,0)--(1.2-0.117739, 1.5);
\end{tikzpicture}

\end{figure}

\subsubsection{Example 3} We now consider domain wall type boundary conditions on the cylinder, or equivalently, periodic domain wall boundary conditions on a real line. The initial conditions are:
\begin{align*}
\rho_0  =  \begin{cases} \rho_L \; & 0< x <x_1 \\ \rho_R \; & x_1<x < L \end{cases}
\end{align*}
\begin{figure}[h]
\begin{tikzpicture}[ scale = 2.4]
\draw[ultra thin, gray, decoration={markings,mark=at position 1 with {\arrow[scale=2]{>}}},
    postaction={decorate}](0,-.7) -- (0,.7);
\draw[ ultra thin, gray, decoration={markings,mark=at position 1 with {\arrow[scale=2]{>}}},
    postaction={decorate}](-.2,0) -- (2.5,0);

\node at  (2.65,0) { $ x $ };
\node at  (0,.85) { $ \rho_0 $ };

\draw (0,.5)--(.6,.5);
\draw (.6,-.5)--(1,-.5);

\draw (1,.5)--(1.6,.5);
\draw (1.6,-.5)--(2,-.5);

\draw (2,.5)--(2.4,.5);

\draw [ultra thin, black!30!white, dashed] (1,-.7)--(1,.7);
\draw [ultra thin, black!30!white,dashed] (2,-.7)--(2,.7);

\draw [ultra thin] (.6,-.075)--(.6,0); \node at (.72,-.1) { \tiny{$ x_1 $} };
\draw [ultra thin] (1.6,-.075)--(1.6,0); 

\draw [ultra thin] (1,-.075)--(1,0);
\draw [ultra thin] (2,-.075)--(2,0);
\node at (1.1,-.1) { \tiny{$ L $} };
\node at (2.15,-.1) { \tiny{$ 2 L $} };

\end{tikzpicture}
 
\end{figure}

In the following analysis, we assume that $ \rho_L =1, \rho_R = -1 $; and $ x_1 > \frac{1}{2} L ( 1 - v) $. The more general initial condition is similar. For these initial conditions, there is both a shock and a rarefaction fan. It is straightforward to solve the differential equations for the propogation of the shock; we omit the details and summarise the result, see Figure \ref{fig:periodicstepcharacteristics}.

For $ y \in [0, y_1] $ the shock moves with constant speed, as the "frozen" characteristics collide. From $ y \in [y_1, y_2] $,the shock curves as the frozen characteristics collide with the rarefaction fan. For $ y > y_2 $, the shock asymptotes to a straight line with slope:
\begin{align*}
\lim_{y \rightarrow \infty} \frac{ds}{dy} = f \left( \rho_{\text{ave}}\right)
\end{align*}
Here $ \rho_{\text{ave}} = \frac{2 x_1 - L}{L} $ is the asymptotic value of $ \rho $, and the average density of $ \rho $ on the cylinder.

A similar picture occurs if $ x_1 < \frac{1}{2} L ( 1- v) $, except that the frozen region corresponding to $ \rho = 1 $ vanishes first. For the critical value of $ x_1 = \frac{1}{2} L ( 1 - v)  $ the frozen regions disappear at the same time $ y_1 = y_2 $.

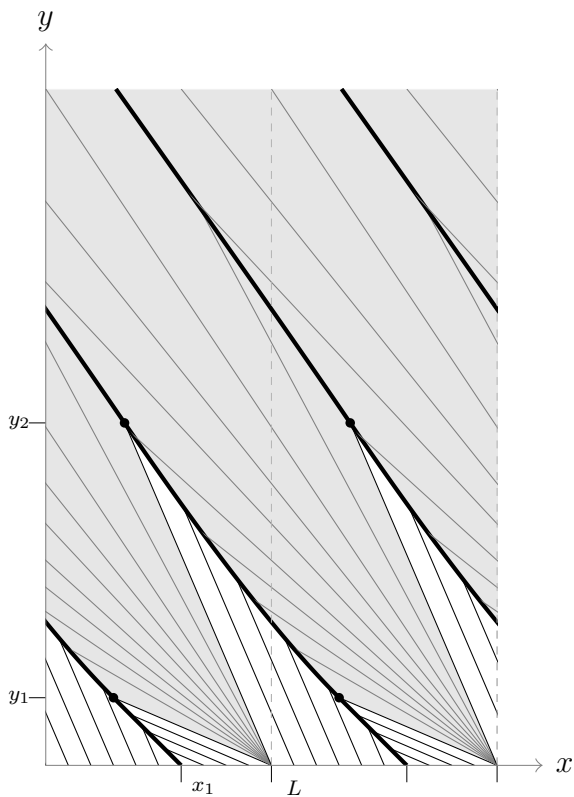
\begin{figure}[H]
\begin{tikzpicture}[ scale = 3]

\draw [ultra thin, black!30!white, dashed] (1,0)--(1,3);
\draw [ultra thin, black!30!white, dashed] (2,0)--(2,3);

\begin{scope}
\clip(0,0) rectangle (2,3);

\begin{scope}
\draw [white, fill = black!10!white] (1,0) -- (1-0.428571*1.51875,1.51875)--(-.68942,3)--(-1.68942,3)--(-0.428571*1.51875,1.51875)-- (.1-0.428571*1.15052,1.15052)--(.2-0.428571*0.833333,0.833333)--  (.4-0.428571*0.35,0.35) -- (.5-0.428571*0.175,0.175) -- (1-2.33333*.3,.3)-- (1,0);

\draw (0,0) -- (-0.428571*1.51875,1.51875);
\draw (.1,0) -- (.1-0.428571*1.15052,1.15052);
\draw (.2,0) -- (.2-0.428571*0.833333,0.833333);
\draw (.3,0) -- (.3-0.428571*0.567188,0.567188);
\draw (.4,0) -- (.4-0.428571*0.35,0.35);
\draw (.5,0) -- (.5-0.428571*0.175,0.175);
\draw (.7,0) -- (.7-2.333331*0.075,0.075);
\draw (.8,0) -- (.8-2.33333*0.15,0.15);
\draw (.9,0) -- (.9-2.33333*0.225,0.225);
\draw (1,0) -- (1-2.33333*.3,.3);

\fill (.3,.3) circle (.6pt);
\fill (-0.650893,1.51875) circle (.6pt);

\draw [domain = 0:.3, smooth, ultra thick] plot ({.6-\x},{\x});
\draw [domain = .3:1.51875, smooth, ultra thick] plot ({0.771429 - 0.625969*sqrt(\x) - 0.428571*\x},{\x});
\draw [domain = 1.51875:3, smooth, ultra thick] plot ({.5 - 0.0867857/ \x - 0.720165*\x},{\x});

\begin{scope} \clip(1,0) -- (1-0.428571*1.51875,1.51875)--(-.68942,3)--(-1.68942,3)--(-0.428571*1.51875,1.51875)-- (.1-0.428571*1.15052,1.15052)--(.2-0.428571*0.833333,0.833333)--  (.4-0.428571*0.35,0.35) -- (.5-0.428571*0.175,0.175) -- (1-2.33333*.3,.3)-- (1,0);

\draw[gray] (1,0)--(-.6,3);
\draw[gray] (1,0)--(-1,3);
\draw[gray] (1,0)--(-1.4,3);
\draw[gray] (1,0)--(-1.8,3);
\draw[gray] (1,0)--(-2.3,3);
\draw[gray] (1,0)--(-2.8,3);
\draw[gray] (1,0)--(-3.4,3);
\draw[gray] (1,0)--(-4.2,3);
\end{scope}

\end{scope}

\begin{scope}[shift={(1,0)}]
\draw [white, fill = black!10!white] (1,0) -- (1-0.428571*1.51875,1.51875)--(-.68942,3)--(-1.68942,3)--(-0.428571*1.51875,1.51875)-- (.1-0.428571*1.15052,1.15052)--(.2-0.428571*0.833333,0.833333)--  (.4-0.428571*0.35,0.35) -- (.5-0.428571*0.175,0.175) -- (1-2.33333*.3,.3)-- (1,0);

\fill (.3,.3) circle (.6pt);
\fill (-0.650893,1.51875) circle (.6pt);

\draw (0,0) -- (-0.428571*1.51875,1.51875);
\draw (.1,0) -- (.1-0.428571*1.15052,1.15052);
\draw (.2,0) -- (.2-0.428571*0.833333,0.833333);
\draw (.3,0) -- (.3-0.428571*0.567188,0.567188);
\draw (.4,0) -- (.4-0.428571*0.35,0.35);
\draw (.5,0) -- (.5-0.428571*0.175,0.175);
\draw (.7,0) -- (.7-2.333331*0.075,0.075);
\draw (.8,0) -- (.8-2.33333*0.15,0.15);
\draw (.9,0) -- (.9-2.33333*0.225,0.225);
\draw (1,0) -- (1-2.33333*.3,.3);

\draw [domain = 0:.3, smooth, ultra thick] plot ({.6-\x},{\x});
\draw [domain = .3:1.51875, smooth, ultra thick] plot ({0.771429 - 0.625969*sqrt(\x) - 0.428571*\x},{\x});
\draw [domain = 1.51875:3, smooth, ultra thick] plot ({.5 - 0.0867857/ \x - 0.720165*\x},{\x});

\begin{scope} \clip(1,0) -- (1-0.428571*1.51875,1.51875)--(-.68942,3)--(-1.68942,3)--(-0.428571*1.51875,1.51875)-- (.1-0.428571*1.15052,1.15052)--(.2-0.428571*0.833333,0.833333)--  (.4-0.428571*0.35,0.35) -- (.5-0.428571*0.175,0.175) -- (1-2.33333*.3,.3)-- (1,0);

\draw[gray] (1,0)--(-.6,3);
\draw[gray] (1,0)--(-1,3);
\draw[gray] (1,0)--(-1.4,3);
\draw[gray] (1,0)--(-1.8,3);
\draw[gray] (1,0)--(-2.3,3);
\draw[gray] (1,0)--(-2.8,3);
\draw[gray] (1,0)--(-3.4,3);
\draw[gray] (1,0)--(-4.2,3);
\end{scope}

\end{scope}

\begin{scope}[shift={(2,0)}]
\draw [white, fill = black!10!white] (1,0) -- (1-0.428571*1.51875,1.51875)--(-.68942,3)--(-1.68942,3)--(-0.428571*1.51875,1.51875)-- (.1-0.428571*1.15052,1.15052)--(.2-0.428571*0.833333,0.833333)--  (.4-0.428571*0.35,0.35) -- (.5-0.428571*0.175,0.175) -- (1-2.33333*.3,.3)-- (1,0);

\fill (.3,.3) circle (.6pt);
\fill (-0.650893,1.51875) circle (.6pt);

\draw (0,0) -- (-0.428571*1.51875,1.51875);
\draw (.1,0) -- (.1-0.428571*1.15052,1.15052);
\draw (.2,0) -- (.2-0.428571*0.833333,0.833333);
\draw (.3,0) -- (.3-0.428571*0.567188,0.567188);
\draw (.4,0) -- (.4-0.428571*0.35,0.35);
\draw (.5,0) -- (.5-0.428571*0.175,0.175);
\draw (.7,0) -- (.7-2.333331*0.075,0.075);
\draw (.8,0) -- (.8-2.33333*0.15,0.15);
\draw (.9,0) -- (.9-2.33333*0.225,0.225);
\draw (1,0) -- (1-2.33333*.3,.3);

\draw [domain = 0:.3, smooth, ultra thick] plot ({.6-\x},{\x});
\draw [domain = .3:1.51875, smooth, ultra thick] plot ({0.771429 - 0.625969*sqrt(\x) - 0.428571*\x},{\x});
\draw [domain = 1.51875:3, smooth, ultra thick] plot ({.5 - 0.0867857/ \x - 0.720165*\x},{\x});

\begin{scope} \clip(1,0) -- (1-0.428571*1.51875,1.51875)--(-.68942,3)--(-1.68942,3)--(-0.428571*1.51875,1.51875)-- (.1-0.428571*1.15052,1.15052)--(.2-0.428571*0.833333,0.833333)--  (.4-0.428571*0.35,0.35) -- (.5-0.428571*0.175,0.175) -- (1-2.33333*.3,.3)-- (1,0);

\draw[gray] (1,0)--(-.6,3);
\draw[gray] (1,0)--(-1,3);
\draw[gray] (1,0)--(-1.4,3);
\draw[gray] (1,0)--(-1.8,3);
\draw[gray] (1,0)--(-2.3,3);
\draw[gray] (1,0)--(-2.8,3);
\draw[gray] (1,0)--(-3.4,3);
\draw[gray] (1,0)--(-4.2,3);
\end{scope}

\end{scope}

\begin{scope}[shift={(3,0)}]
\draw [white, fill = black!10!white] (1,0) -- (1-0.428571*1.51875,1.51875)--(-.68942,3)--(-1.68942,3)--(-0.428571*1.51875,1.51875)-- (.1-0.428571*1.15052,1.15052)--(.2-0.428571*0.833333,0.833333)--  (.4-0.428571*0.35,0.35) -- (.5-0.428571*0.175,0.175) -- (1-2.33333*.3,.3)-- (1,0);

\draw (0,0) -- (-0.428571*1.51875,1.51875);
\draw (.1,0) -- (.1-0.428571*1.15052,1.15052);
\draw (.2,0) -- (.2-0.428571*0.833333,0.833333);
\draw (.3,0) -- (.3-0.428571*0.567188,0.567188);
\draw (.4,0) -- (.4-0.428571*0.35,0.35);
\draw (.5,0) -- (.5-0.428571*0.175,0.175);
\draw (.7,0) -- (.7-2.333331*0.075,0.075);
\draw (.8,0) -- (.8-2.33333*0.15,0.15);
\draw (.9,0) -- (.9-2.33333*0.225,0.225);
\draw (1,0) -- (1-2.33333*.3,.3);

\draw [domain = 0:.3, smooth, ultra thick] plot ({.6-\x},{\x});
\draw [domain = .3:1.51875, smooth, ultra thick] plot ({0.771429 - 0.625969*sqrt(\x) - 0.428571*\x},{\x});
\draw [domain = 1.51875:3, smooth, ultra thick] plot ({.5 - 0.0867857/ \x - 0.720165*\x},{\x});

\begin{scope} \clip(1,0) -- (1-0.428571*1.51875,1.51875)--(-.68942,3)--(-1.68942,3)--(-0.428571*1.51875,1.51875)-- (.1-0.428571*1.15052,1.15052)--(.2-0.428571*0.833333,0.833333)--  (.4-0.428571*0.35,0.35) -- (.5-0.428571*0.175,0.175) -- (1-2.33333*.3,.3)-- (1,0);

\draw[gray] (1,0)--(-.6,3);
\draw[gray] (1,0)--(-1,3);
\draw[gray] (1,0)--(-1.4,3);
\draw[gray] (1,0)--(-1.8,3);
\draw[gray] (1,0)--(-2.3,3);
\draw[gray] (1,0)--(-2.8,3);
\draw[gray] (1,0)--(-3.4,3);
\draw[gray] (1,0)--(-4.2,3);
\end{scope}

\end{scope}

\end{scope}

\draw [ultra thin] (.6,-.075)--(.6,0); \node at (.7,-.1) { \tiny{$ x_1 $} };
\draw [ultra thin] (1.6,-.075)--(1.6,0); 

\draw [ultra thin] (1,-.075)--(1,0);
\draw [ultra thin] (2,-.075)--(2,0);
\node at (1.1,-.1) { \tiny{$ L $} };

\draw [ultra thin, black!30!white, dashed] (1,0)--(1,3);
\draw [ultra thin, black!30!white, dashed] (2,0)--(2,3);

\draw [ultra thin] (-.075,.3)--(0,.3); \node at (-.1,.3) {\tiny{$ y_1 $ } };
\draw [ultra thin] (-.075,1.51875)--(0,1.51875);  \node at (-.1,1.51875) {\tiny{$ y_2 $ } };

\draw[ultra thin, gray, decoration={markings,mark=at position 1 with {\arrow[scale=2]{>}}},
    postaction={decorate}](0,0) -- (0,3.2);
\draw[ ultra thin, gray, decoration={markings,mark=at position 1 with {\arrow[scale=2]{>}}},
    postaction={decorate}](0,0) -- (2.2,0);

\node at  (2.3,0) { $ x $ };
\node at  (0,3.3) { $ y $ };
\end{tikzpicture}
\caption{Step initial conditions on the cylinder.} \label{fig:periodicstepcharacteristics}
\end{figure}

\appendix

\section{Stochasticity of the transfer matrix}\label{sec:StoTrans}

\begin{lemma}\label{lem:sum-trans}
With the stochastic weights (\ref{eq:stochasticweights}) and fixed state $ \alpha $:
\begin{align*}
\sum_{\beta \in \mathcal{S}} \langle \beta, T_M \alpha\rangle = 1 + b_1^{M-m } b_2^{m}
\end{align*}
where $ m = m(\alpha) $.
\end{lemma}
\begin{proof}
For a state $ \alpha $, we will write $ \alpha_i $ for the position of the $ i$th particle for $ i = 1, \cdots, m $, with the convention that $ \alpha_{m+1} = \alpha_1 + M $. For two states $ \alpha $ and $ \beta $, we associate the 6-vertex configuration $ (\alpha, \beta) $ with the $ i$th particle going from $ \alpha_i $ to $ \beta_i $. The ice rule implies that a valid configuration $(\alpha, \beta) $ must satisfy $ \alpha_i \leq \beta_i \leq \alpha_{i+1} $, and $ \beta_i \neq \beta_{i+1} $ for all $ i $.

$\Delta=\text{cosh}(\eta)$, $H=-V=\pm \eta/2$, $\gamma=\mp(u+\eta)$.

Let us define the weight of a particle's path as:
\begin{align}\label{eq:pathweight}
w(\alpha_i, \beta_i) = \begin{cases} b_2 &\mbox{if } \alpha_i = \beta_i \\
c_1 c_2 \; b_1^{ \beta_i - \alpha_i -1 } & \mbox{if } \beta_i < \alpha_{i+1} \\
b_1^{ \beta_i - \alpha_i -1} & \mbox{if } \beta_i = \alpha_{i+1} \end{cases}
\end{align}
(Only the third case isn't immediately obvious from 6-vertex weights.) It is straightforward to check that the weight of any valid configuration $ (\alpha, \beta) $ is the product of path weights:
\begin{align*}
W(\alpha, \beta) = \prod_{i=1}^m  w(\alpha_i, \beta_i)
\end{align*}

We will compute the sum over all configurations by summing each $ \beta_i $ from $ \alpha_i $ to $ \alpha_{i+1} $, and afterward subtracting the overcounted terms corresponding to invalid configurations (when $ \beta_i = \beta_{i+1} $  for some $ i $).
\\

\noindent \emph{Example of two Particles:} The sum is easily computed by geometric series:
\begin{align*}
\sum_{\beta_{i} = \alpha_{i}}^{\alpha_{i+1}} w(\alpha_i, \beta_i) &= b_2 + \sum_{\beta_{i} = \alpha_{i}+1}^{\alpha_{i+1}-1} c_1 c_2 \; b_1^{ \beta_i - \alpha_i -1 } + b_1^{\alpha_{i+1} - \alpha_i - 1} \\
&= b_2 + c_1 c_2 \frac{( 1 - b_1)^{\alpha_{i+1} - \alpha_{i} - 1}}{1 - b_1} + b_1^{\alpha_{i+1} - \alpha_{i}  -1} \\
&= b_2 + (1 - b_2)( 1 - b_1)^{\alpha_{i+1} - \alpha_{i} - 1} + b_1^{\alpha_{i+1} - \alpha_{i}  -1} \\
&= 1 + b_1 b_2^{\alpha_{i+1} - \alpha_{i}- 1}
\end{align*}
Here we used $ c_1 = 1 - b_1 $ and $ c_2 = 1 - b_2 $. So we have:
\begin{align*}
&= (1 + b_1 b_2^{\alpha_{2} - \alpha_{1}- 1}) ( 1 + b_1 b_2^{\alpha_{1}+ M - \alpha_{2}- 1}) \\
&= 1 +  b_1 b_2^{\alpha_{2} - \alpha_{1}- 1} + b_1 b_2^{\alpha_{1}+M - \alpha_{2}- 1} + b_1^2 b_2^{M - 2}
\end{align*}
The middle terms are precisely the weights assigned to the overlap states, $ w(\alpha_1, \alpha_2 )  w(\alpha_2,\alpha_2)  $ and  $ w(\alpha_1, \alpha_1) w(\alpha_2, \alpha_1+M) $. Subtracting these gives:
\begin{align*}
\sum_{\beta} W(\alpha, \beta) = ( 1 + b_1 b_2^{M-1})
\end{align*}

\noindent \emph{General case:} We can compute in the same way as above that:
\begin{align} \label{eq:prodpaths}
\prod_{i=1}^m \sum_{\beta_i= \alpha_i}^{\alpha_{i+1}}w(\alpha_i, \beta_i) = (1 + b_1 b_2^{\alpha_2 - \alpha_1 - 1}) (1 + b_1 b_2^{\alpha_3 - \alpha_2 - 1}) \cdots (1 + b_1 b_2^{\alpha_1 + M  - \alpha_n - 1})
\end{align}
We can organize the subtraction of invalid configurations as follows. Writing (\ref{eq:prodpaths}) as
\begin{align*}
(A_1 + B_1)(A_2 + B_2)\cdots(A_m+B_m)
\end{align*}
and expanding, a term in the expansion containing $ A_i B_{i+1} $ for any $ i $ represents sums of weights of states with the $ i $th particle and $ (i+1) $th particle overlapping, and invalid case. It follows that what remains is $ A_1 \cdots A_m+  B_1 \cdots B_m $. We then have:
\begin{align*}
A_1 \cdots A_m + B_1 \cdots B_m &= 1+ (b_1 b_2^{\alpha_2 - \alpha_1 - 1})(b_1 b_2^{\alpha_3 - \alpha_2 - 1}) \cdots (b_2^{\alpha_1 + M  - \alpha_m - 1}) \\
&= 1 + b_1^n b_2^{M - m}
\end{align*}

\end{proof}

\section{The Surface Tension Function for $ \Delta > 1 $} \label{sec:Surf}The free energy and surface tension of the 6-vertex model in the $ \Delta > 1 $ regime was studied in detail in \cite{BS} \cite{No}. We summarize briefly the computations and relevant results.

For the critical value of magnetic field $ H_0 = \pm \eta / 2 = \pm \cosh^{-1}(\Delta) / 2 $ and any magnetization (slope), the roots of the Bethe ansatz equation lie on a closed contour in the complex plane and can be solved exactly by Fourier transform. Near the critical value of $ H_0 $, the Bethe ansatz equations can be solved perturbatively to yield the perturbative expansions for the free energy $ f $ and the surface tension $ \sigma $.

In the phase diagram of the free energy $ f $, the exact solutions correspond to the conical singularities of $ f(H,V) $, at $ (H_0, V_0) = (\pm \eta /2 , \mp \eta/2) $ (see Figure (\ref{fig:phasediag1})). In the surface tension $ \sigma $, the conical points correspond to critical lines (see Figure  (\ref{fig:phasediag2})); each direction of approach to the conical point maps (by Legendre transform) to a slope $ (s,t) $ on the critical line. The critical lines in $ \sigma $ are given by the equation:
\begin{align} \label{eq:critline}
s = \frac{t \pm \tanh(u+\eta)}{1 \pm \tanh(u+\eta) t }
\end{align}

Near the critical line, the partial derivatives of the $ \sigma $ can be computed using the perturbative expressions given in \cite{BS}. A long but straightforward calculation, which we omit, yields:

\begin{lemma}
For slopes $ (s,t) $ on the critical lines (\ref{eq:critline}):
\begin{align*}
\sigma_{11}(s,t) \sigma_{22}(s,t) - \sigma_{12}(s,t)^2 = 0
\end{align*}
\end{lemma}

\begin{figure}[h]
\subfloat[Subfigure 1 list of figures text][]{
\includegraphics[scale=.8]{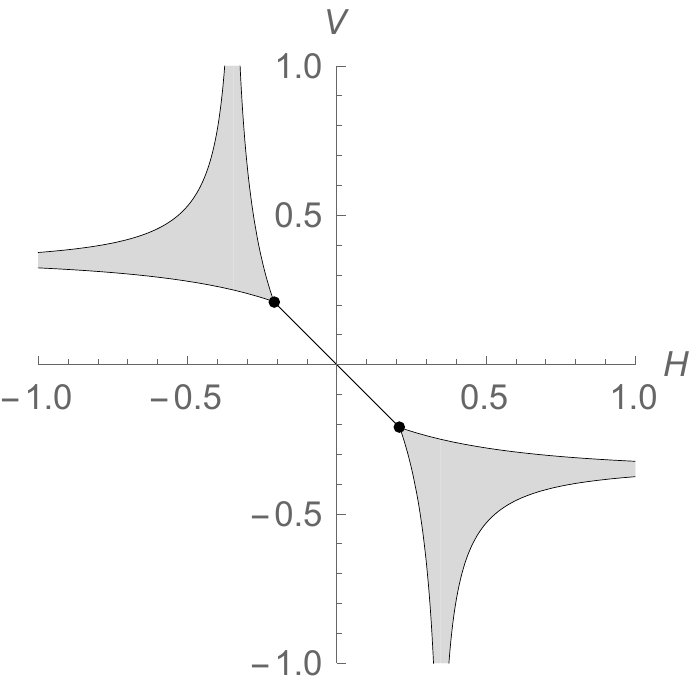}
\label{fig:phasediag1}}
\; \; \; \; \; \;
\subfloat[Subfigure 2 list of figures text][]{
\includegraphics[scale=.8]{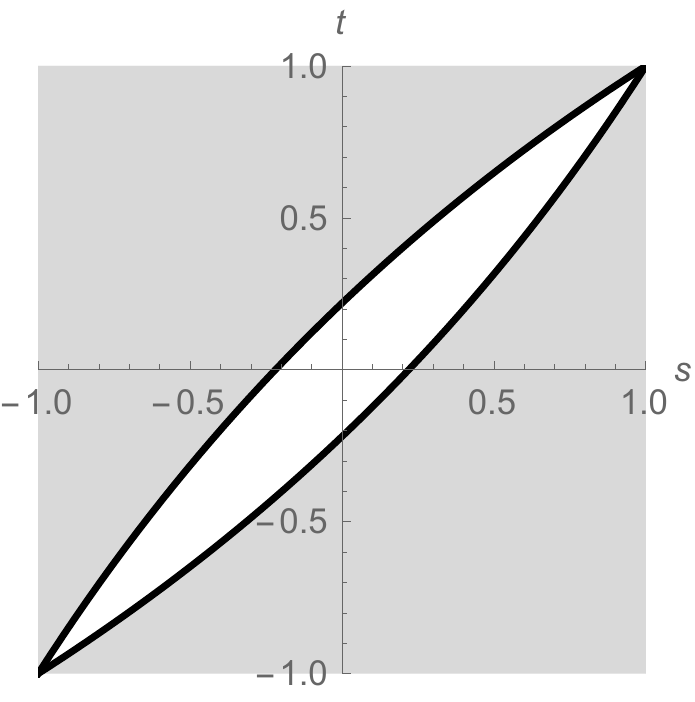}
\label{fig:phasediag2}}
\caption{(A) The phase diagram in the $ (H, V) $ plane. The shaded regions correspond to the liquid phase; the white regions are frozen. The dots are the location of the conical singularities in $ f $. (B) The phase diagram in the $ (s,t) $ plane. The shaded regions correspond to the liquid phase. The bold lines correspond to the conical points. There are no pure states with slopes between these critical lines.}
\label{fig:globfig}
\end{figure}

\begin{figure}
\includegraphics[scale=.6]{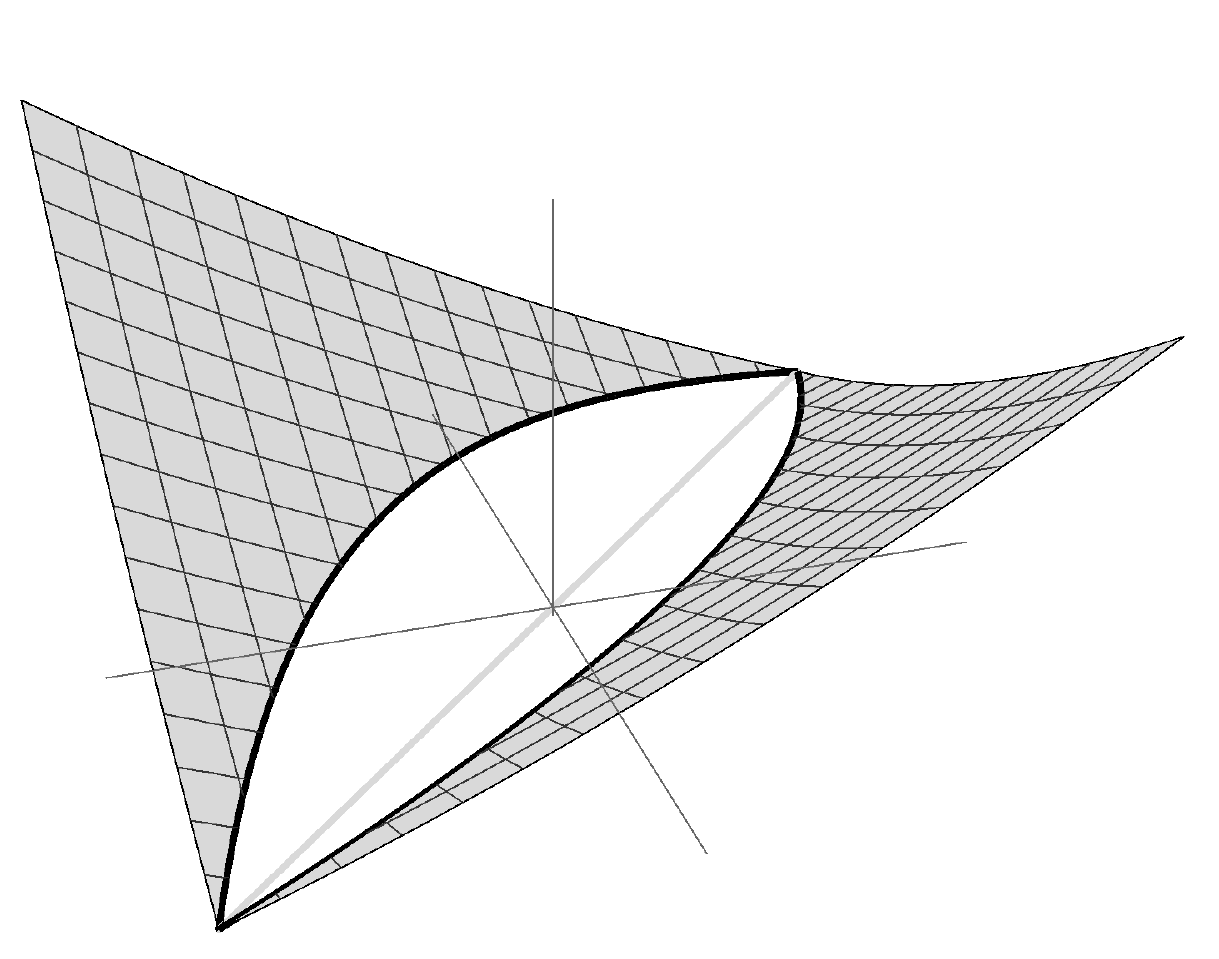}
\caption{A Sketch of the $ \sigma $, reproduced from \cite{BS}.}\label{fig:sigmasketch}
\end{figure}

\end{document}